\documentclass[letterpaper, amsfonts, amssymb, amsmath, pra, aps, reprint, onecolumn, showkeys, nofootinbib, 11pt]{revtex4-2}
\usepackage[margin=1in]{geometry}

\usepackage{setspace}
\usepackage{array}
\usepackage{amsmath}
\usepackage{amsfonts}
\usepackage{amsthm}
\usepackage{algorithm}
\usepackage{algorithmic}
\usepackage{bm, bbm}
\usepackage{amssymb}
\usepackage{booktabs} 
\usepackage{enumerate}
\usepackage{graphicx}
\usepackage{multirow}
\usepackage{physics}
\usepackage{setspace}
\usepackage{threeparttable}
\usepackage{upgreek}
\usepackage{subcaption}
\usepackage{xcolor}
\usepackage{hyperref}
\usepackage{natbib}

\renewcommand{\mu}{\upmu}
\hypersetup{
    colorlinks = true,
    allcolors = black
}

\graphicspath{{figures/}}

\newcommand{\beginsupplement}{%
    \setcounter{table}{0}
    \renewcommand{\thetable}{S\arabic{table}}%
    \setcounter{figure}{0}
    \renewcommand{\thefigure}{S\arabic{figure}}%
    \renewcommand{\thesubsection}{\arabic{subsection}}
 }

\newtheorem{theorem}{\bf Theorem}
\newtheorem{lem}{\bf Lemma}



\begin{document}
\begin{titlepage}
\title{\Large\bfseries A Non-Nested Quantum Search for TSP: Explicit Oracle Construction Beyond the $2^n$ Barrier}

\author{Xujun Bai$^{1, 2}$}
\author{Yun Shang$^{1, 3, \dag, *}$}
\author{Honghong Lin$^{1, 2}$}

\affiliation{
$^1$Institute of Mathematics, Academy of Mathematics and Systems Science, Chinese Academy of Sciences, Beijing 100190, China \\
$^2$School of Mathematical Sciences, University of Chinese Academy of Sciences, Chinese Academy of Sciences, Beijing, 100049, China \\
$^3$State Key Laboratory of Mathematical Sciences, Academy of Mathematics and Systems Science, Chinese Academy of Sciences, Beijing, 100190, China
}

\let\thefootnote\relax\footnotetext{$^\dag$ shangyun@amss.ac.cn}
\let\thefootnote\relax\footnotetext{$^*$ Corresponding authors }

\begin{abstract}
    The traveling salesman problem (TSP) is a classic NP-hard problem. Held--Karp dynamic programming~\cite{held1962dynamic, bellman1962dynamic} solves it exactly in $O(n^2 2^n)$ time. Ambainis et al.~\cite{ambainis2019quantum}\ (SODA~2019) proposed a hybrid classical-quantum dynamic programming that improved the query complexity to $O^*(1.727^n)$. However, their algorithm uses nested quantum minimum finding, which complicates coherent circuit implementation, particularly the handling of threshold updates within nested subroutines. 
    To address this difficulty, we design a quantum divide-and-conquer framework without any recursive process. We overcome the oracle-construction bottleneck by preparing a set partition state, requiring $O(n^2)$ gates and $O(n)$ depth. Under assumed QRAM model, the structured state preparation enables polynomial-depth oracle evaluation, yielding an overall running time of $O^*(1.866^n)$ for TSP. Qiskit simulations for $n=6$ and 7, using QFT-based data loading in place of QRAM, yield empirical optimal-solution probabilities of $98.9\%$ and $100\%$, respectively, in fixed-threshold amplitude-amplification tests.
\end{abstract}

\maketitle
\end{titlepage}

\section{Introduction}
The traveling salesman problem (TSP) asks: given $n$ cities and pairwise distances, find the shortest tour that visits each city exactly once and returns to the start. TSP has been thoroughly studied in terms of both exact and approximation algorithms~\cite{laporte1992traveling, chauhan2012survey}. The best known classical exact algorithm for general TSP is the Held-Karp dynamic programming~\cite{held1962dynamic, bellman1962dynamic}, which runs in $O(n^2 2^n)$ time. Despite decades of effort, no classical algorithm for general TSP has broken this $2^n$ barrier. A longstanding open question is whether quantum computing can improve on Held-Karp. 
Quantum computing offers a quadratic speedup for unstructured search: Grover's algorithm~\cite{grover1996fast} can find the marked element among $N$ candidates using $O(\sqrt{N})$ queries. A direct application to TSP would search over all $(n-1)!$ Hamiltonian cycles, giving $O^*(\sqrt{(n-1)!}) \approx 2^{\frac{n}{2}\log_2 n}$ queries~\cite{bai2025quantum}, which is far worse than $2^n$. The challenge is to combine classical dynamic programming with quantum search to surpass $O^*(2^n)$. A central obstacle is that existing quantum speedups for TSP rely on the black-box oracle, leaving the explicit circuit implementation of the nested oracle, together with its resource costs, to be established. This paper focuses on this obstacle: how to design a quantum TSP algorithm whose oracle can be explicitly constructed and implemented.

We design a quantum divide-and-conquer framework parameterized by the number of subsets $k$ and their sizes $m_1\geq\dots\geq m_k$. The idea is to precompute shortest paths within subsets of size at most $m_1$ via classical dynamic programming, then search over all set $k$-partitions with labeled origins and ends in a single coherent quantum process using quantum minimum finding. {\bf Lemma~\ref{thm2}} shows that the optimal tour can always be expressed as a concatenation of such paths, so the search is guaranteed to find it. 
The framework balances classical preprocessing against quantum search through the number and sizes of the subsets. 
Treating $k$ and $m_i$ as free parameters allows us to optimize the classical-quantum balance.
We fully characterize this trade-off: $k=2$ yields no speedup, $k=3$ achieves $O^*(1.890^n)$, \textbf{$k=4$ attains the global optimum $O^*(1.866^n)$}, and $k>4$ cannot be better than $k=4$.

\subsection*{Our contribution and techniques}
Our main contributions and techniques are summarized as follows.

{\bf Quantum divide-and-conquer and the optimal query complexity.}  
We develop a quantum divide-and-conquer strategy to search over all feasible set partition schemes. Each scheme is characterized by the number of subsets $k$ and the proportions $\alpha_i=m_i/n$ satisfying
\begin{equation*}
    1>\alpha_1\geq\alpha_2\geq\dots\geq\alpha_{k-1}\geq1-\sum_{i=1}^{k-1}\alpha_i>0.
\end{equation*}
We can first find all shortest paths from an origin to an end for all subsets of size at most $\alpha_1n$ by classical dynamic programming, and then apply quantum divide-and-conquer to search for the shortest Hamiltonian cycle. The query complexities of the classical and quantum parts are $O^*(2^{nH(\alpha_1)})$ and 
\begin{equation*}
    O^*\left(\sqrt{\binom{n}{\alpha_1n}\binom{(1-\alpha_1)n}{\alpha_2n}\dots\binom{(1-\sum_{i=1}^{k-2}\alpha_i)n}{\alpha_{k-1}n}}\right),
\end{equation*}
respectively. We prove that the total query complexity reaches the optimal value $O^*(1.865666\dots^n)$ when $k=4$ and $\alpha_1=\alpha_2=\alpha_3=0.315742\dots$, which is obtained when the classical query complexity equals the quantum query complexity. ({\bf Theorem~\ref{thm3}})

{\bf Structured state preparation.}
Query complexity alone does not guarantee a genuine speedup. To mark a complete solution, the oracle must compute the lengths of Hamiltonian cycles determined by set partitions and the classically precomputed shortest paths of the subsets; yet classical precomputation of these cycle lengths would nullify the quantum speedup, so the evaluation must be done in quantum parallel, which renders the oracle construction inherently nontrivial. We resolve this obstacle by encoding the set partitions explicitly in a structured quantum state. We utilize the Dicke state preparation to produce the uniform superposition state over all set $k$-partitions with labeled origins and ends, which we call the \emph{set partition state}. The circuit requires $O(n^2)$ gates and $O(n)$ depth with at most constant ancillary qubits. ({\bf Theorem~\ref{thm: setPartitionState}}) Leveraging structured state preparation and quantum random access memory (QRAM), we achieve an overall running time of $O^*(1.866^n)$ under the computational model specified in Section~\ref{secIIIC}.

Qiskit simulations for $n=6$ and 7, using a shortcut of quantum Fourier transform (QFT) in place of QRAM, yield empirical optimal-solution probabilities of $98.9\%$ and $100\%$, respectively, in fixed-threshold amplitude-amplification tests.

\subsection*{Related work}
We summarize previous results on TSP and our result in {\bf Table~\ref{tab: complexity_summary}}.
\begin{table}[h]
    \centering
    \caption{Summary of results for general TSP. All complexities are expressed in the $O^*(\cdot)$ sense (ignoring polynomial factors).~\cite{caroppo2026quantum} achieves $O^*(1.315\dots^n)$ QRAM by trading time for space. Like~\cite{ambainis2019quantum}, it recursively calls quantum minimum finding and does not provide an explicit oracle construction. $M$ depends on the encoding accuracy.}
    \label{tab: complexity_summary}
    \setlength{\tabcolsep}{3mm}
    \begin{tabular}{l c c c c c}
        \toprule[1.2pt]
        & & & \multicolumn{3}{c}{Oracle} \\
        \cmidrule(lr){4-6}
        Algorithm & Query complexity & QRAM & Black-box? & Nested? & Depth \\
        \midrule
        Held-Karp~\cite{held1962dynamic, bellman1962dynamic} & $2^{n}$ & -- & --  & -- & --\\[2pt]
        Ambainis et al.~\cite{ambainis2019quantum} & $1.727^n$ & $1.727^n$   & Yes & Yes & --\\[2pt]
        Caroppo et al.~\cite{caroppo2026quantum} & $1.859^n$ & $1.315^n$    & Yes & Yes & --\\[2pt]
        \textbf{This work} & $1.866^n$ & $1.866^n$   & No  & No & $O(n^3+nM)$ \\
        \bottomrule[1.2pt]
    \end{tabular}
\end{table}

Ambainis et al.~\cite{ambainis2019quantum} proposed a hybrid algorithm that improved the query complexity from Held--Karp's $O^*(2^n)$ to $O^*(1.727^n)$. Caroppo et al.~\cite{caroppo2026quantum} improved the space complexity by trading it for time, showing that TSP can be solved quantumly in $O^*(1.859^n)$ time and $O^*(1.315^n)$ QRAM space. However, their algorithm recursively calls quantum minimum finding and does not provide a specific circuit implementation for the recursive call of the oracle. 
For such nested schemes, an explicit circuit-level resource analysis must account for coherent subroutine execution, the construction of the oracle, and uncomputation. 

Constructing and executing the oracle is the central practical challenge for any quantum search algorithm. Prior works on oracles have explored robust quantum minimum finding~\cite{quek2020robust}, oracles via phase kick-back~\cite{lee2016generalised, ossorio2023generalisation}, weak measurements~\cite{andres2022weakly}, recursive oracle expansion~\cite{burke2025deterministic}, and spinorial representations~\cite{mascarenhas2026quantum}, but constructing an oracle for a hard practical problem remains highly problem-dependent and challenging.
We address this difficulty by shifting the burden from oracle design to structured state preparation, using an architecture that separates index and value registers~\cite{bai2025quantum}: divide the qubits into index registers and value registers, prepare a uniform superposition over the combinatorial objects of interest in the index registers, and compute the corresponding cost into the value registers. By preparing the superposition state directly, we avoid the need for a bijection between sequential indices and combinatorial objects inside the oracle.

Quantum superposition states of combinatorial objects appear in quantum search~\cite{bai2025quantum}, quantum walks~\cite{marsh2020combinatorial}, and variational algorithms~\cite{bartschi2020grover} for combinatorial optimization. Examples range from Dicke states for vertex cover~\cite{jiang2025dicke} to Hamiltonian cycle superpositions for TSP~\cite{bai2025quantum} and permutation states for graph comparison~\cite{chiew2019graph}. In this work, we utilize the Dicke state preparation to produce uniform superpositions over set partitions with labels. This structured state preparation enables explicit oracle construction while preserving the quantum speedup under the assumed QRAM model.

\subsection*{Paper organization}
Section~\ref{secII} develops the quantum divide-and-conquer framework and proves its correctness (Lemma~\ref{thm2}), and establishes the optimal query complexity (Theorem~\ref{thm3}).
Section~\ref{secIII} explains the role of structured state preparation, introduces the set partition state preparation circuit (Theorem~\ref{thm: setPartitionState}), and describes the oracle implementation with QRAM.
Section~\ref{secIV} reports Qiskit simulation results on small instances.
Section~\ref{secV} discusses the applicability and limitations of our framework.

\section{Quantum divide-and-conquer} \label{secII}

\subsection{Revisit the algorithm for TSP in~\cite{ambainis2019quantum}} \label{secIIA}
Let $A=(\omega_{ij})_{n\times n}$ be the weighted adjacency matrix of the complete undirected graph $G=(V,\ E)$ and $f(u,\ S)$ denote the cost of the shortest route that starts from $u\notin S$ and passes through each vertex in $S\subset V$ exactly once. The best known classical dynamic programming algorithm for TSP is provided as {\bf Algorithm~\ref{alg: cDP-TSP}}~\cite{held1962dynamic, bellman1962dynamic}, the time complexity of which is $O(n^2 2^n)$.

\begin{algorithm}[H]
\begin{algorithmic}[1]
\item[\ \ \ \ \textbf{Input:}] The weighted adjacency matrix $A=(\omega_{ij})_{n\times n}$. A special vertex that the salesman starts from, here we select $0\in V$.
\item[\ \ \ \ \textbf{Output:}] The shortest Hamiltonian cycle that starts and ends at the vertex $0$ and its cost.
\STATE For all vertices $u\neq0$, set $f(u,\ \emptyset)=\omega_{0u}$.
\FOR {$k=1$ to $n-1$}
    \IF {$k<n-1$}
        \STATE $\forall S\subset V\setminus\{0\}$ where $|S|=k$, $\forall u\in V\setminus(S\cup\{0\})$, calculate $f(u,\ S)=\min_{t\in S}\{f(t,\ S\setminus\{t\})+\omega_{ut}\}$. Record the minimizer $t^*$.
    \ELSE
        \STATE For $S=V\setminus\{0\}$, calculate the cost of the shortest Hamiltonian cycle of the graph $G$: $f(0,\ S)=\min_{t\in S}\{f(t,\ S\setminus\{t\})+\omega_{t0}\}$. Record the minimizer $t^*$.
    \ENDIF
\ENDFOR
\STATE Trace back through the DP table to retrieve the shortest Hamiltonian cycle.
\end{algorithmic}
\caption{The classical dynamic programming for TSP.}\label{alg: cDP-TSP}
\end{algorithm}

The recursive formulas for TSP used in~\cite{ambainis2019quantum} are 
\begin{equation}
    f(S,\ u,\ v)=\min_{\substack{t\in N(u)\cap S \\ t\neq v}}\{f(S\setminus\{u\},\ t,\ v)+\omega_{ut}\},\ \ \ f(\{u,v\},\ u,\ v)=\omega_{uv}, \label{equ: ambainis1}
\end{equation}
\begin{equation}
    f(S,\ u,\ v)=\min_{\substack{X\subset S,\ |X|=k \\ u\in X,\ v\notin X}}\min_{\substack{t\in X \\ t\neq u}}\{f(X,\ u,\ t)+f((S\setminus X)\cup\{t\},\ t,\ v)\}, \label{equ: ambainis2}
\end{equation}

\begin{equation}
    f(V)=\min_{\substack{S\subset V \\ |S|=\frac{n}{2}}}\min_{\substack{u,\ v\in S \\ u\neq v}}\{f(S,\ u,\ v)+f((V\setminus S)\cup\{u,\ v\},\ v,\ u)\}, \label{equ: ambainis3}
\end{equation}
where $|S|\ge3$, $f(S,\ u,\ v)$ denotes the minimum cost of the routes that pass through all vertices in $S$ exactly once starting at $u\in S$ and ending at $v\in S$, $f(V)$ denotes the cost of the shortest Hamiltonian cycle of $G$, and $N(u)$ is the set of neighbors of $u$. The algorithm proposed in~\cite{ambainis2019quantum} is provided as {\bf Algorithm~\ref{alg: hDP-TSP}}.

\begin{algorithm}[H]
\begin{algorithmic}[1]
\item[\ \ \ \ \textbf{Input:}] The weighted adjacency matrix $A=(\omega_{ij})_{n\times n}$. Let $\alpha\leq1/2$ to ensure $\alpha n/4\leq(1-\alpha)n/4$.
\item[\ \ \ \ \textbf{Output:}] The shortest Hamiltonian cycle and its cost.
\STATE Calculate $f(S,\ u,\ v)$ for all $|S|\leq(1-\alpha)n/4$ classically using the formula (\ref{equ: ambainis1}) and store them in memory.
\STATE Run quantum minimum finding over all subsets $S\subset V$ where $|S|=n/2$ to calculate the answer $f(V)$ using the formula (\ref{equ: ambainis3}). To calculate $f(S,\ u,\ v)$ for all $|S|=n/2$, run quantum minimum finding using the formula (\ref{equ: ambainis2}) with $k=n/4$. To calculate $f(S,\ u,\ v)$ for all $|S|=n/4$, run quantum minimum finding using the formula (\ref{equ: ambainis2}) with $k=\alpha n/4$. For $|S|=\alpha n/4$ or $|S|=(1 - \alpha)n/4$, we know $f(S,\ u,\ v)$ from the classical process.
\end{algorithmic}
\caption{The hybrid dynamic programming for TSP.}\label{alg: hDP-TSP}
\end{algorithm}

The classical step 1 of {\bf Algorithm~\ref{alg: hDP-TSP}} requires the complexity of
\begin{equation}
    O^*\left(\binom{n}{\leq (1-\alpha)n/4}\right)=O^*(2^{H(\frac{1-\alpha}{4})n}).
    \label{equ: com1}
\end{equation}
Counting the number of distinct subset selection at step 2 of {\bf Algorithm~\ref{alg: hDP-TSP}}, the authors of~\cite{ambainis2019quantum} derive the query complexity of the quantum part as
\begin{equation}
    O^*\left(\sqrt{\binom{n}{n/2}\binom{n/2}{n/4}\binom{n/4}{\alpha n/4}}\right)=O^*(2^{\frac{1}{2}(1+\frac{1}{2}+\frac{H(\alpha)}{4})n}).
    \label{equ: com2}
\end{equation}
In formulas (\ref{equ: com1}) and (\ref{equ: com2}), $O^*$ indicates ignoring the polynomial factors, and the entropy approximation is used:
\begin{equation}
    \begin{aligned}
        \binom{n}{k}&\leq2^{nH(k/n)}, \\
        \binom{n}{\leq k}&=\sum_{i=0}^{k}\binom{n}{i}\leq2^{nH(k/n)},
        \label{equ: entropy-approx}
    \end{aligned}
\end{equation}
where $k\leq n/2$ and $H(\epsilon)=-[\epsilon\log_2\epsilon+(1-\epsilon)\log_2(1-\epsilon)],\ \epsilon\in[0,\ 1]$ is the binary entropy. The authors of~\cite{ambainis2019quantum} demonstrated that the overall complexity is minimized when the complexities of the classical part and the quantum part are equal, which gives the optimal $\alpha\approx 0.055362$ and the corresponding complexity of {\bf Algorithm~\ref{alg: hDP-TSP}} is $O^*(1.727391\dots^n)$. This improved on the best known classical complexity of $O(n^2 2^n)$ given by {\bf Algorithm~\ref{alg: cDP-TSP}}.

\subsection{Quantum divide-and-conquer} \label{secIIB}
{\bf Notation.} Let $A=(\omega_{ij})_{n\times n}$ be the weighted adjacency matrix of the complete undirected graph $G=(V,\ E)$, where $|V|=n$. For convenience, throughout this paper, we assume nonnegative integer edge weights satisfying $0\leq\omega_{ij}\leq W$, where $C=\lceil\log_2(W+1)\rceil$ is polynomially bounded in $n$. Lengths of Hamiltonian cycles and thresholds for quantum minimum finding are integers in $[0,nW]$. Polynomial factors in both $n$ and $C$ are suppressed in the $O^*$ notation. 
Given $k\in\mathbb{Z}^+$ and sizes $m_1,\ m_2,\ ...,\ m_k\in\mathbb{Z}^+$ where $k\leq n/2,\ m_i\ge2$ and $\sum_{i=1}^km_i=n$, we divide the $n$-vertex set $V$ into $k$ disjoint subsets $S_i$ of sizes $m_1,\ m_2,\ ...,\ m_k$ respectively to obtain a $k$-partition and select two arbitrary vertices from each subset as origin and end, denoted by $u_i,\ v_i$. 
We can label each subset, e.g., A, B, C, ..., and encode each label with $\lceil\log_2k\rceil$ bits. For each vertex $v\in V$, we define the concatenation of three binary strings $l_vo_ve_v$, where $l_v$ is the code of the label of the subset to which $v$ belongs; $o_v=1$ if $v$ is the origin of its subset, and $o_v=0$ otherwise; $e_v=1$ if $v$ is the end of its subset, and $e_v=0$ otherwise. The concatenation $p=\prod_v(l_vo_ve_v)$ is the encoding of the above $k$-partition with origins and ends. We collect encodings corresponding to all possible $k$-partitions with origins and ends into a set $P$.

\begin{lem} \label{thm2}
    For each $k$-partition, let $L_i$ be the shortest path from $u_i$ to $v_i$ that visits every vertex in $S_i$ exactly once. Concatenating these paths and adding edges $(v_1,u_2),\ (v_2,u_3),\ ...,\ (v_{k-1},u_k),\ (v_k,u_1)$ yields a cycle corresponding to this $k$-partition. Let $P_H$ be the set of all such cycles formed by all $k$-partitions with sizes $m_1,\ m_2,\ ...,\ m_k$ and all possible selections of $u_i,\ v_i$ for each $S_i$. Then, the shortest Hamiltonian cycle belongs to $P_H$.
\end{lem}

\begin{proof}
    \begin{figure}[H]
        \centering
        \includegraphics[width=.8\linewidth]{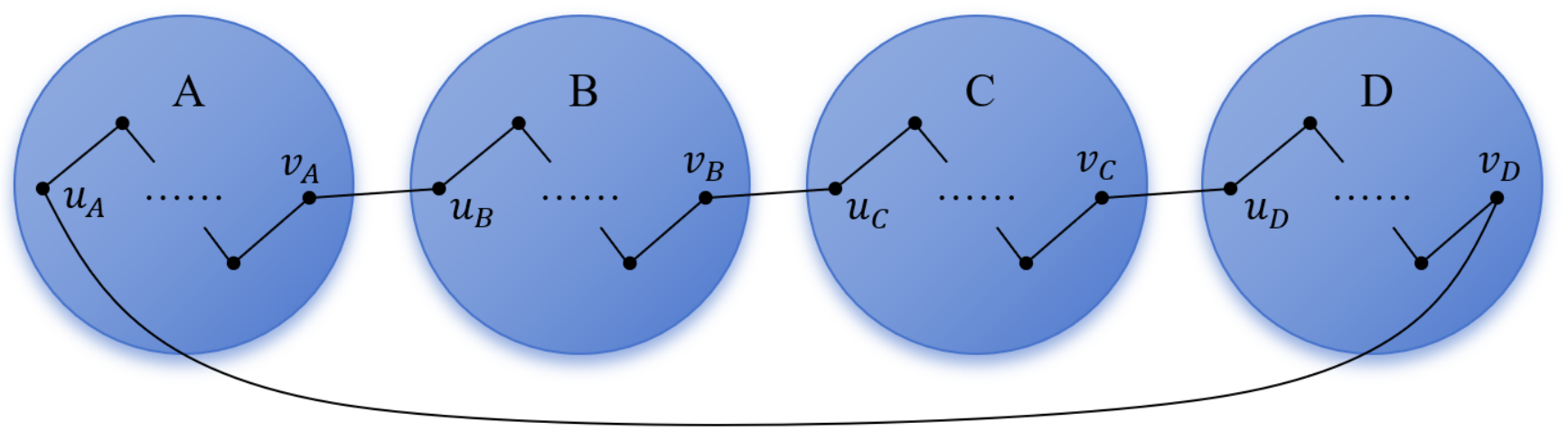}

        \caption{Example of $4$-partition.}
        \label{fig: TSP-4-partition}
    \end{figure}
    {\bf Figure~\ref{fig: TSP-4-partition}} shows an example of $4$-partition. Let the shortest Hamiltonian cycle be $w_1\to w_2\to\ ...\ \to w_n \to w_1$. We can start from $w_1$ and select $m_1,\ ...,\ m_k$ vertices sequentially along the circle to form subsets $S_1,\ ...,\ S_k$. Assume the origin and end of the cycle in $S_i$ are $u_i$ and $v_i$ respectively. Then, the part of the cycle in $S_i$ must be the shortest path from $u_i$ to $v_i$ that passes through all vertices exactly once in $S_i$. Otherwise, we can replace this path with the shortest path to obtain a shorter Hamiltonian cycle, leading to a contradiction. Therefore, the shortest Hamiltonian cycle belongs to $P_H$.
\end{proof}
{\bf Lemma~\ref{thm2}} suggests that we can first compute all shortest paths of subsets by classical dynamic programming and then search for the shortest Hamiltonian cycle over the components of the uniform superposition state encoding all elements of $P$. Assume $m_1\geq m_2\geq\ ...\geq m_k\geq2$ w.l.o.g., then it suffices to apply dynamic programming to all subsets with a size no more than $m_1$.

\begin{algorithm}[H]
\begin{algorithmic}[1]
\item[\ \ \ \ \textbf{Input:}] The weighted adjacency matrix $A=(\omega_{ij})_{n\times n}$. Let $n>m_1\geq m_2\geq\ ...\geq m_k\geq2$.
\item[\ \ \ \ \textbf{Output:}] The shortest Hamiltonian cycle and its cost.
\STATE Calculate $f(S_i,\ u_i,\ v_i)$ for all $|S_i|\leq m_1$ classically using (\ref{equ: ambainis1}).
\STATE Prepare the uniform superposition state $\ket{\psi}$ encoding all elements of $P$, which we call the set partition state.
\STATE Run quantum minimum finding over the components of the superposition state $\ket{\psi}$ to get the shortest Hamiltonian cycle, where the oracle can query classical data about $f(S_i,\ u_i,\ v_i)$ for $|S_i|\leq m_1$.
\end{algorithmic}
\caption{The quantum divide-and-conquer algorithm for TSP.}\label{alg: hdc-TSP}
\end{algorithm}

The query complexity of the classical part of {\bf Algorithm~\ref{alg: hdc-TSP}} is 
\begin{equation}
    O^*\left(\sum_{i=1}^{m_1}\binom{n}{i}\right)=O^*(2^{nH(\frac{m_1}{n})})\ \ \text{if}\ \ m_1\leq n/2,\ O^*(2^n)\ \ \text{otherwise}.
    \label{equ: com3}
\end{equation}
The size of the quantum search space is $|P|=(\frac{n!}{\prod_{i=1}^km_i!})[\prod_{i=1}^km_i(m_i-1)]$. For fixed $k$, the endpoint-selection factor $\prod_{i=1}^km_i(m_i-1)$ is polynomial in $n$ and is absorbed into $O^*$. This yields the query complexity of the quantum part 
\begin{equation}
    O(\sqrt{|P|})=O^*\left(\sqrt{\binom{n}{m_1}\binom{n-m_1}{m_2}\dots\binom{n-\sum_{i=1}^{k-2}m_i}{m_{k-1}}}\right)=O^*\left(\sqrt{\frac{n!}{\prod_{i=1}^km_i!}}\right).
    \label{equ: com4}
\end{equation}

\subsection{Optimal query complexity} \label{secIIC}
The optimal complexity analysis builds on Eqs. (\ref{equ: com3}) and (\ref{equ: com4}), where the asymptotic complexity formulas that suppress endpoint-selection factors are stated for fixed $k$.
Let $m_i=\alpha_in,\ 0<\alpha_i<1,\ 1\leq i\leq k-1$ and $m_k=(1-\sum_{i=1}^{k-1}\alpha_i)n$. Then the query complexity of the classical part of {\bf Algorithm~\ref{alg: hdc-TSP}} is 
\begin{equation}
    O^*(2^{nH(\alpha_1)})\ \ \text{if}\ \ \alpha_1\leq\frac{1}{2},\ O^*(2^n)\ \ \text{otherwise}.
    \label{equ: com5}
\end{equation}
The query complexity of the quantum part is 
\begin{align}
    &O^*\left(\sqrt{\binom{n}{\alpha_1n}\binom{(1-\alpha_1)n}{\alpha_2n}\dots\binom{(1-\sum_{i=1}^{k-2}\alpha_i)n}{\alpha_{k-1}n}}\right)\notag\\
    =& O^*\left(2^{\frac{n}{2}(H(\alpha_1)+(1-\alpha_1)H(\frac{\alpha_2}{1-\alpha_1})+\ \dots\ +(1-\sum_{i=1}^{k-2}\alpha_i)H(\frac{\alpha_{k-1}}{1-\sum_{i=1}^{k-2}\alpha_i}))}\right).
    \label{equ: com6}
\end{align}

Let $f(\alpha_1)=H(\alpha_1)$ and $g_k(\vec{\alpha})=\frac{1}{2}(H(\alpha_1)+(1-\alpha_1)H(\frac{\alpha_2}{1-\alpha_1})+\ \dots\ +(1-\sum_{i=1}^{k-2}\alpha_i)H(\frac{\alpha_{k-1}}{1-\sum_{i=1}^{k-2}\alpha_i}))$ denote the exponential parts of the classical and quantum complexities, respectively.

\begin{theorem} \label{thm3}
    The optimal query complexity of {\bf Algorithm~\ref{alg: hdc-TSP}} is $O^*(1.865666...^n)$ and this optimal value is obtained when $k=4$ and $\alpha_1=\alpha_2=\alpha_3=0.315742...$
\end{theorem}

\begin{proof}
    First, we show that $g_k(\vec{\alpha})=-\frac{1}{2}[\sum_{i=1}^{k-1}\alpha_i\log_2\alpha_i+(1-\sum_{i=1}^{k-1}\alpha_i)\log_2(1-\sum_{i=1}^{k-1}\alpha_i)]$.
    When $k=2$, $g_2(\alpha_1)=\frac{1}{2}H(\alpha_1)=-\frac{1}{2}[\alpha_1\log_2\alpha_1+(1-\alpha_1)\log_2(1-\alpha_1)]$.
    Assume the proposition holds for $k=t-1$. Then, for $k=t$ we have:
    
    \begin{align}
        g_t(\vec{\alpha})=&g_{t-1}(\vec{\alpha}) + \frac{1}{2}(1-\sum_{i=1}^{t-2}\alpha_i)H(\frac{\alpha_{t-1}}{1-\sum_{i=1}^{t-2}\alpha_i})\notag\\
        =&-\frac{1}{2}[\sum_{i=1}^{t-2}\alpha_i\log_2\alpha_i+(1-\sum_{i=1}^{t-2}\alpha_i)\log_2(1-\sum_{i=1}^{t-2}\alpha_i)]\notag\\
        &-\frac{1}{2}(1-\sum_{i=1}^{t-2}\alpha_i)(\frac{\alpha_{t-1}}{1-\sum_{i=1}^{t-2}\alpha_i}\log_2\frac{\alpha_{t-1}}{1-\sum_{i=1}^{t-2}\alpha_i}+\frac{1-\sum_{i=1}^{t-1}\alpha_i}{1-\sum_{i=1}^{t-2}\alpha_i}\log_2\frac{1-\sum_{i=1}^{t-1}\alpha_i}{1-\sum_{i=1}^{t-2}\alpha_i}) \notag\\
        =&-\frac{1}{2}[\sum_{i=1}^{t-1}\alpha_i\log_2\alpha_i+(1-\sum_{i=1}^{t-1}\alpha_i)\log_2(1-\sum_{i=1}^{t-1}\alpha_i)].
        \label{equ: g_k}
    \end{align}
    Hence, the proposition holds for any $k\geq2$ by induction.

    Next, we analyze the monotonicity of the functions $f(\alpha_1)$ and $g_k(\vec{\alpha})$. 
    If $\alpha_1\ge1/2$, the classical preprocessing already requires $O^*(2^n)$ complexity. Hence, any parameter choice that improves on Held–Karp must satisfy $k>2$ and $\alpha_1<1/2$. In this regime, $f(\alpha_1)$ is monotonically increasing. 
    By (\ref{equ: g_k}), the gradient of $g_k(\vec{\alpha})$ is 
    \begin{align}
        \nabla g_k(\vec{\alpha})=&[-\frac{1}{2\ln2}(\ln\alpha_j+1-\ln(1-\sum_{i=1}^{k-1}\alpha_i)-1)]_{j=1}^{k-1}\notag\\
        =&(-\frac{1}{2}\log_2\frac{\alpha_j}{1-\sum_{i=1}^{k-1}\alpha_i})_{j=1}^{k-1}.
        \label{equ: grad}
    \end{align}
    We have $\frac{1}{2}>\alpha_1\geq\ ...\ \geq\alpha_{k-1}\geq1-\sum_{i=1}^{k-1}\alpha_i>0$ for $n>m_1\geq m_2\geq\ ...\geq m_k>0$, deriving that $\frac{\alpha_j}{1-\sum_{i=1}^{k-1}\alpha_i}\geq1$ (and $\alpha_1\geq\frac{1}{k}$). Hence, $\nabla g_k(\vec{\alpha})\leq0$ and $g_k(\vec{\alpha})$ is monotonically decreasing in each of its variables. 
    The value of $g_k(\vec{\alpha})$ can be decreased by increasing the value of any variable, until $\sum_{i=1}^{k-1}\alpha_i\to1$. (This is always possible for $k\geq3$; for instance, we can take $\alpha_1,\alpha_2\to\frac{1}{2}$ for $k=3$.) In this limit, we have $\alpha_{k-1}=1-\sum_{i=1}^{k-2}\alpha_i$ and the value of $g_{k}$ becomes 
    \begin{equation}
        -\frac{1}{2}\sum_{i=1}^{k-1}\alpha_i\log_2\alpha_i=-\frac{1}{2}[\sum_{i=1}^{k-2}\alpha_i\log_2\alpha_i+(1-\sum_{i=1}^{k-2}\alpha_i)\log_2(1-\sum_{i=1}^{k-2}\alpha_i)],
        \label{equ: degrade}
    \end{equation}
    which degenerates to the case of $g_{k-1}$ by (\ref{equ: g_k}).

    The total query complexity is dominated by $\max\{f(\alpha_1),\ g_k(\vec{\alpha})\}$. For $\alpha_1\leq\frac{1}{k-1}$, we define  
    \begin{equation}
        g^*_k(\alpha_1)=g_k(\vec{\alpha})|_{\alpha_1=\ ...\ =\alpha_{k-1}}=-\frac{1}{2}[(k-1)\alpha_1\log_2\alpha_1+(1-(k-1)\alpha_1)\log_2(1-(k-1)\alpha_1)].
        \label{equ: g_opt}
    \end{equation}
    Then we have $g^*_k(\alpha_1)\ge g^*_k(\frac{1}{k-1})=\lim_{\alpha_1\to\frac{1}{k-1}}g^*_k(\alpha_1)=\frac{1}{2}\log_2(k-1)$ and 
    \begin{equation}
        g^*_k(\frac{1}{k-1})\geq1,\ \forall k\geq5. \label{equ: g_bound}
    \end{equation}
    We now show that it suffices to consider $k<5$ for the optimal total complexity. We temporarily optimize over the closure of the feasible region by allowing some $\alpha_i=0$. These boundary points represent limits of valid positive partitions and are used only to analyze lower bounds. We fix $k$ and $\alpha_1$.
    \begin{enumerate}
        \item[(i)] If $\alpha_1<\frac{1}{k-1}$, the value of $g_k(\vec{\alpha})$ can be decreased by increasing the value of other variables, until $\alpha_1=\alpha_2=\ ...\ =\alpha_{k-1}$ (because the variables have been sorted in descending order and $\sum_{i=1}^{k-1}\alpha_i=(k-1)\alpha_1<1$). Hence, we have $g_k(\vec{\alpha})\ge g_k^*(\alpha_1)> g^*_k(\frac{1}{k-1})$; and if $k\ge5$, the complexity cannot be better than classical {\bf Algorithm~\ref{alg: cDP-TSP}} by (\ref{equ: g_bound}).
        \item[(ii)] If $\alpha_1\ge\frac{1}{k-1}$, the value of $g_k(\vec{\alpha})$ can be decreased by increasing the value of other variables, until $\sum_{i=1}^{k-1}\alpha_i=1$ and the value of $g_k$ is equal to the value of $g_{k-1}$ by (\ref{equ: degrade}). This process can be continued until, at some $r$, $\alpha_1=\ ...\ =\alpha_{r-1},\ \sum_{i=1}^{r-1}\alpha_i<1$ and the value of $g_k$ decreases to the value of $g_r$. Now we have $0<\alpha_r\leq\alpha_{r-1},\ \alpha_j=0$ for $r<j<k$ and 
        \begin{equation}
            \sum_{i=1}^{r-1}\alpha_i=(r-1)\alpha_1<1=\sum_{i=1}^{r}\alpha_i\leq r\alpha_1. \label{equ: bound_con}
        \end{equation}
        (\ref{equ: bound_con}) implies $\alpha_1<\frac{1}{r-1}$ and $\frac{1}{\alpha_1}\leq r<\frac{1}{\alpha_1}+1$. Hence, $r=\lceil\frac{1}{\alpha_1}\rceil$ and $g_r^*(\alpha_1)>g_r^*(\frac{1}{r-1})$. Therefore, by (\ref{equ: g_bound}), any $\alpha_1$ that can outperform the classical {\bf Algorithm~\ref{alg: cDP-TSP}} must satisfy $\lceil\frac{1}{\alpha_1}\rceil<5$, i.e., $\alpha_1\ge\frac{1}{4}$. Moreover, in this case, since $f$ is fixed for a given $\alpha_1$, and the optimization of $g_k$ eventually drives its value down to that of $g_r$, it suffices to optimize $g_r$ where $r<5$.
    \end{enumerate} 

    Finally, we calculate the optimal query complexity. For $k=3$, we have $\alpha_1\ge\frac{1}{3}$. It is easy to see that $g_3$ attains its global maximum of $0.792481\dots$ at $\alpha_1=\alpha_2=\frac{1}{3}$, which is below the minimum value of $f$, i.e., $f(\frac{1}{3})=0.918295\dots$ The total query complexity is dominated by classical part, i.e., $O^*(2^{0.918295\dots n})=O^*(1.889881\dots^n)$. For $k=4$, we have $\alpha_1\ge\frac{1}{4}$. If $\alpha_1<\frac{1}{3}$, we have $g_4(\vec{\alpha})\ge g_4^*(\alpha_1)$. Since $g_4^*$ is monotonically decreasing while $f$ is monotonically increasing, the optimal complexity is attained when $f(\alpha_1)=g_4^*(\alpha_1)$, where numerical evaluation gives $f(0.315742\ldots) = g_4^*(0.315742\ldots) = 0.899691\ldots$
    If $\alpha_1\ge\frac{1}{3}$, the complexity cannot be lower because $f(\alpha_1)\ge f(\frac{1}{3})$.

    Therefore, the optimal query complexity of {\bf Algorithm~\ref{alg: hdc-TSP}} is $O^*(2^{0.899691...n})=O^*(1.865666...^n)$.
\end{proof}

The trade-off behind the optimal $k$ can be understood as follows. The classical part precomputes the shortest paths within subsets of size at most $\alpha_1 n$, with cost $O^*(2^{nH(\alpha_1)})$. The quantum part searches over all $k$-partitions with labels of origins and ends, with cost $O^*(\sqrt{n!/(\prod_{i=1}^k m_i!)})$. As $k$ increases, each subset shrinks, making the classical pre-computation cheaper. However, the number of $k$-partitions grows rapidly with $k$, and the quantum search cost rises accordingly. At $k=2$, there is no speedup because the classical cost is $O^*(2^n)$. At $k=3$, the complexity is entirely dominated by the classical cost, yielding the optimal exponents $0.918295\dots$ At $k=4$, there is a balance point where the classical and quantum cost are equal, yielding the optimal exponents $0.899691\ldots$ Beyond $k=4$, the growth in the number of partitions quickly exceeds the quadratic Grover speedup, and the overall complexity cannot be better.

While $k=4$ is optimal, the numerical difference between $k=3$ ($1.889881\ldots^n$) and $k=4$ ($1.865666\ldots^n$) is modest, and $k=3$ requires only three subsets, which simplifies the algorithm. For small to moderate $n$, $k=3$ may be preferable in practice, and our Qiskit simulations ({\bf Section~\ref{secIV}}) use $k=3$.

\section{Implementation of the oracle} \label{secIII}
\subsection{Role of structured state preparation} \label{secIIIA}
Structured state preparation provides a direct route to polynomial-depth oracle evaluation under the QRAM model, while avoiding the exponential preprocessing overhead of a lookup table over complete candidate solutions.

The oracle accesses the classically precomputed data to form complete solutions. A straightforward approach prepares a uniform superposition state $\sum_{i=0}^{|P|-1}\frac{1}{\sqrt{|P|}}\ket{i}$ via Hadamard gates, then uses a lookup table that maps each index $i$ to a set partition with labels. The oracle then queries the corresponding classical data about the length of the route. This approach requires a bijection from sequential indices to set partitions with labels, and the oracle must encode this bijection, which makes the oracle construction prohibitively difficult.
We avoid this difficulty by shifting the burden from oracle design to state preparation: we encode $P$ explicitly in a structured quantum state. We utilize the method for preparing Dicke states to directly prepare the uniform superposition state that encodes set partitions with labels, which we call the set partition state. Although this approach does not fully utilize the encoding space of quantum states, it avoids the need to design and encode a bijection in the oracle by using the structural nature of the set partition state.

The necessity of structured state preparation becomes even clearer when the oracle is implemented via QRAM. (We take the bucket-brigade QRAM as an example. A brief introduction about the bucket-brigade QRAM can be found in \hyperref[app1]{\bf Appendix A}.)
With the trivial encoding, which prepares the state $\sum_{i=0}^{|P|-1}\frac{1}{\sqrt{|P|}}\ket{i}\ket{0}^{\otimes M}$ via $H^{\otimes}$ gates, the QRAM query process is as follows:

Our TSP solver first precomputes the shortest paths for subsets $S_i$ of size at most $\alpha_1n$ by classical dynamic programming. Then we classically calculate the lengths of Hamiltonian cycles determined by set partitions with labels and the corresponding DP paths by querying the precomputed shortest paths $f(S_i,u_i,v_i)$ and the weights of connecting edges $\omega_{vu}$. We subtract the threshold from these lengths to obtain the values of 

\begin{equation}
    f(S_A,0,v_A)+\omega_{v_Au_B}+f(S_B,u_B,v_B)+\omega_{v_Bu_C}+f(S_C,u_C,v_C)+\omega_{v_Cu_D}+f(S_D,u_D,v_D)+\omega_{v_D0}-C_T.\notag
\end{equation}
and save these values as a classical array $Arr=\{(i,\ v_i):0\leq i\leq|P|-1\}$ according to two's complement form. The elements of this array are assigned to the leaf nodes of the QRAM tree, and the binary strings corresponding to the paths from the root to these leaf nodes encode the indices of the array. In this way, the oracle can access classical data via QRAM. Therefore, we can obtain the state $\sum_{i=0}^{|P|-1}\frac{1}{\sqrt{|P|}}\ket{i}\ket{v_i}$ from the state $\sum_{i=0}^{|P|-1}\frac{1}{\sqrt{|P|}}\ket{i}\ket{0}^{\otimes M}$ through the QRAM query. Next, the oracle of our TSP solver applies $Z$-gate on the sign qubit and then resets the value register to its initial state $\ket{0}^{\otimes M}$ using the inverse of the QRAM query operation.

Under the optimal parameters, we have $|P|=O^*(2^{2nH(\alpha_1)})$ because the classical query complexity $O^*(2^{nH(\alpha_1)})$ and the quantum query complexity $O(\sqrt{|P|})$ are equal. Creating the classical array $Arr=\{(i,\ v_i):0\leq i\leq|P|-1\}$ and assigning its elements to the leaf nodes of the QRAM tree therefore costs $O^*(2^{2nH(\alpha_1)})=O^*(1.865666\dots^{2n})=O^*(3.480709\dots^n)$, erasing the quantum advantage. We must compute values related to $P$ in a quantum manner, and can only query the classical precomputed DP table of length $O^*(2^{nH(\alpha_1)})$, which motivates our use of the structured state to explicitly encode $P$.
In the next section, we will construct the required state. We build on the Dicke state preparation framework of~\cite{bartschi2019deterministic}.

\subsection{Structured quantum state} \label{secIIIB}
{\bf The Dicke state.} 
The Dicke state~\cite{dicke1954coherence}, defined as the uniform superposition of all $n$-qubit states $\ket{x}$ with Hamming weight $w(x)=k$, i.e., 
\begin{equation}
    \ket{D^n_k}=\sum_{x\in\{0,1\}^n,\ w(x)=k}\frac{1}{\sqrt{\binom{n}{k}}}\ket{x},
    \label{equ:dicke1}
\end{equation}
can be used to index the subsets of the same size. Bartschi et al.~\cite{bartschi2019deterministic} provided a deterministic method to efficiently prepare the Dicke state without any ancillary qubits, which requires $O(kn)$ quantum gates and a depth of $O(n)$. This method is based on the recursive formula
\begin{equation}
    \ket{D^n_k}=\sqrt{\frac{n-k}{n}}\ket{D^{n-1}_k}\ket{0}+\sqrt{\frac{k}{n}}\ket{D^{n-1}_{k-1}}\ket{1},
    \label{equ:dicke2}
\end{equation}
which starts from the initial state $\ket{0}^{\otimes(n-k)}\ket{1}^{\otimes k}$ and can be achieved by applying an appropriate $RY$ gate. 
The quantum circuits for preparing Dicke states consist of two basic structures, shown in {\bf Figure~\ref{fig: dicke-structure}}.
\begin{figure}[H]
    \centering
    \begin{subfigure}{0.45\textwidth}
        \centering
        \includegraphics[width=\textwidth]{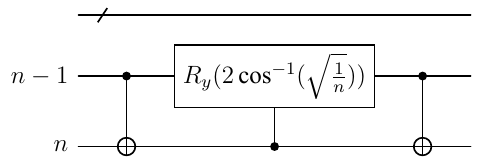}
        \caption{The first structure.}
        \label{fig: dicke-sub1}
    \end{subfigure}
    \hfill
    \begin{subfigure}{0.45\textwidth}
        \centering
        \includegraphics[width=\textwidth]{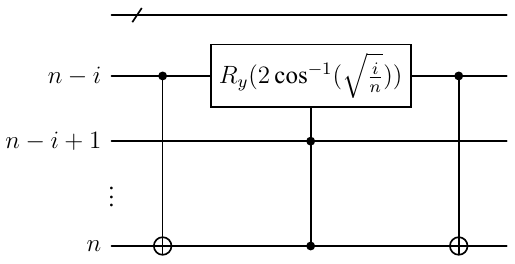}
        \caption{The second structure.}
        \label{fig: dicke-sub2}
    \end{subfigure}
    \caption{Basic structures of the quantum circuits for preparing Dicke states, which are used to realize the operation (\ref{equ:dicke2}).}
    \label{fig: dicke-structure}
\end{figure}

{\bf Figure~\ref{fig: Dicke4q}} provides a quantum circuit to prepare $\ket{D^4_2}$, in which the $RY(2\cos^{-1}\sqrt{\frac{i}{m}})$ gate can transform the state $\ket{0}$ into the state $\sqrt{\frac{i}{m}}\ket{0}+\sqrt{\frac{m-i}{m}}\ket{1}$ and there are three modules with $m=4,\ 3,\ 2$. In each module, the circuit scans upward from the current highest-order qubit (the $m$-th qubit counting from the top). Upon locating the first $0$, a $RY$ gate is applied to generate two branches. For the branch where the qubit becomes $\ket{1}$, the corresponding $\ket{1}$ at the current highest-order qubit is flipped.
As shown in {\bf Figure~\ref{fig: DickeTree}}, the three layers of the binary tree (from left to right) correspond to the actions of the three modules. The underlined digits indicate the corresponding qubits that will not be modified in subsequent operations, and the digit immediately preceding an underlined digit represents the current highest-order qubit. The leaf nodes of the tree represent $6$ components of $\ket{D^4_2}$.
\begin{figure}[H]
    \centering
    \includegraphics[width=.9\linewidth]{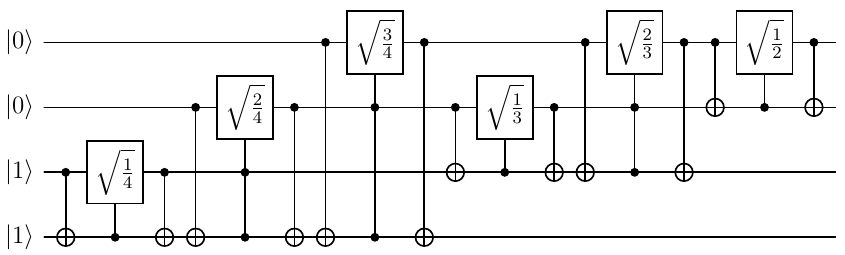}

    \caption{Quantum circuit for preparing Dicke state of $4$ qubits.}
    \label{fig: Dicke4q}
\end{figure}
 
\begin{figure}[H]
    \centering
    \includegraphics[width=.4\linewidth]{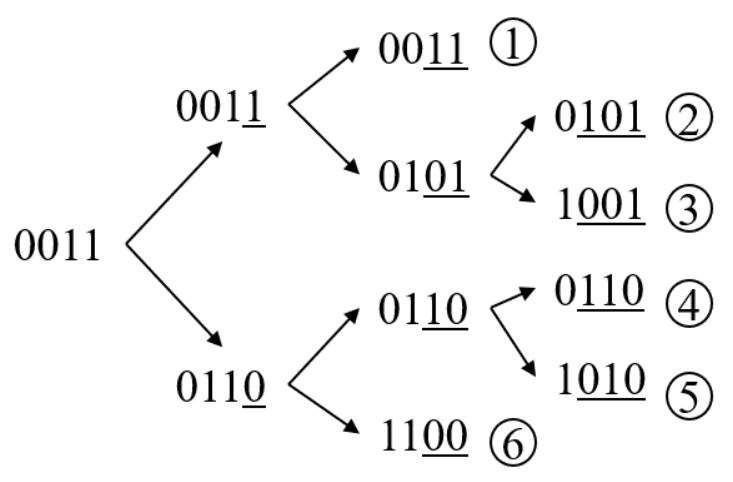}

    \caption{The binary tree produced in preparing $\ket{D^4_2}$.}
    \label{fig: DickeTree}
\end{figure}

The Dicke state can be used to design quantum algorithms for some combinatorial optimization problems, such as $k$-vertex cover problem~\cite{jiang2025dicke}. Using the binary tree perspective, we will leverage the Dicke state preparation as a building block to construct the set partition state.

{\bf The set partition state.} 
We can distinguish $4$ subsets with labels A, B, C, D and use binary strings $l=00,\ 01,\ 10,\ 11$ to represent them respectively. To identify the origin and the end of the shortest path in a subset, we use four qubits to encode a vertex of the entire set $V$. A total of $4n$ qubits is required. Among the four qubits encoding each vertex $i$, the first two qubits $\ket{l_i}$ indicate which subset the vertex belongs to, and the last two qubits $\ket{o_ie_i}$ indicate the origin and end: If this vertex is the origin, then $o_i=1,\ e_i=0$. If this vertex is the end, then $o_i=0,\ e_i=1$. For other cases, $o_i=0,\ e_i=0$.
Based on this encoding method, our set partition state is 
\begin{equation}
    \sum_{p\in P}\frac{1}{\sqrt{|P|}}\ket{p}=\sum_{p\in P}\frac{1}{\sqrt{|P|}}\otimes_{i=0}^{n-1}\ket{l_io_ie_i}, \label{equ: setPartitionState}
\end{equation}
where $P$ and $(l_io_ie_i)$ have been defined in {\bf Section~\ref{secIIB}}. Each $p\in P$ determines a candidate Hamiltonian cycle after the corresponding DP paths are selected. 

\begin{theorem} \label{thm: setPartitionState}
    The set partition state can be prepared with the gate complexity of $O(n^2)$ and the depth of $O(n)$ using at most constant ancillary qubits, and the algorithm is provided as {\bf Algorithm~\ref{alg: setPartition}}.
\end{theorem}

\begin{figure}[H]
    \centering
    \includegraphics[width=.5\linewidth]{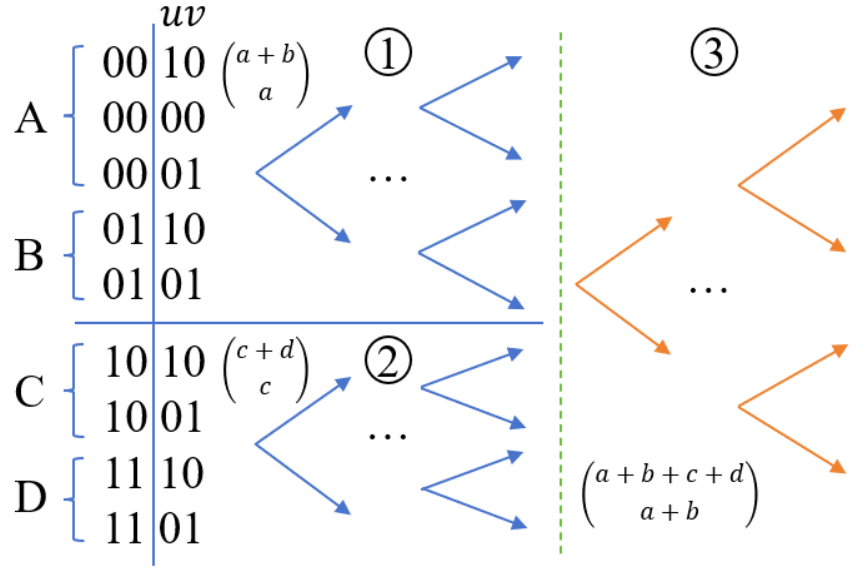}

    \caption{The binary tree produced in preparing a $4$-partition state with size of $3,\ 2,\ 2,\ 2$ ($n=9$). A total of $36$ qubits is used, where $4$ qubits for each vertex, the first and second columns of qubits encode labels of subsets, and the third and fourth columns of qubits indicate the origin and end.}
    \label{fig: SetPartitionTree}
\end{figure}

\begin{algorithm}[H]
\begin{algorithmic}[1]
\item[\ \ \ \ \textbf{Input:}] Sizes of subsets $a,\ b,\ c,\ d\geq2$, where $a+b+c+d=n$. The initial state $\ket{0}^{\otimes4n}$.
\item[\ \ \ \ \textbf{Output:}] The uniform superposition state of set partitions with labels: $\sum_{p\in P}\frac{1}{\sqrt{|P|}}\otimes_{i=0}^{n-1}\ket{l_io_ie_i}$, where every $4$-qubit encodes one vertex.
\STATE Apply $X$ gates to the first three columns of qubits to initialize them such that the first column of qubits is set to $\ket{0}^{\otimes(a+b)}\ket{1}^{\otimes(c+d)}$, the second column of qubits is set to $\ket{0}^{\otimes a}\ket{1}^{\otimes b}\ket{0}^{\otimes c}\ket{1}^{\otimes d}$, and the third column of qubits is set to $\ket{0}^{\otimes(a-2)}\ket{1}^{\otimes 2}\ket{0}^{\otimes(b-2)}\ket{1}^{\otimes 2}\ket{0}^{\otimes(c-2)}\ket{1}^{\otimes 2}\ket{0}^{\otimes(d-2)}\ket{1}^{\otimes 2}$.
\STATE Apply the Dicke state preparation algorithms of $D^a_2,\ D^b_2,\ D^c_2,\ D^d_2$ to the third column of qubits for $4$ subsets, respectively.
\STATE Apply the basic modules in {\bf Figure~\ref{fig: SplitModule}} to all possible $\binom{x}{2}$ two-qubit positions for $4$ subsets, respectively, where $x=a,\ b,\ c,\ d$.
\STATE Apply the Dicke state $D^{a+b}_b$ preparation algorithm with basic modified structures such as examples in {\bf Figure~\ref{fig: par-structure}} to the first $4(a+b)$ qubits, where the $RY$ gates act on the second column of qubits.
\STATE Apply the Dicke state $D^{c+d}_d$ preparation algorithm with basic modified structures such as examples in {\bf Figure~\ref{fig: par-structure}} to the last $4(c+d)$ qubits, where the $RY$ gates act on the second column of qubits.
\STATE Apply the Dicke state $D^{n}_{c+d}$ preparation algorithm with basic modified structures such as examples in {\bf Figure~\ref{fig: par-structure}} to all $4n$ qubits, where the $RY$ gates act on the first column of qubits.
\end{algorithmic}
\caption{The set partition state preparation algorithm for $k=4$. We illustrate each step according to the column configuration of vertices ($4$ columns of qubits), as shown in {\bf Figure~\ref{fig: SetPartitionTree}}.}\label{alg: setPartition}
\end{algorithm}

\begin{proof}
We adopt the 4-qubit-per-vertex encoding described before Theorem~\ref{thm: setPartitionState}.

The step 1 of {\bf Algorithm~\ref{alg: setPartition}} creates the initial state for the Dicke state preparation algorithms, requiring $O(n)$ $X$ gates. Steps 2 and 3 of {\bf Algorithm~\ref{alg: setPartition}} encode the information of the origins and the ends. We apply the Dicke state preparation algorithms four times on the third column of qubits for subsets A, B, C, D, respectively, to get the state $\ket{D^a_2}\ket{D^b_2}\ket{D^c_2}\ket{D^d_2}$, requiring the gate complexity of $O(4\cdot2n)=O(n)$~\cite{bartschi2019deterministic}. Because one vertex can be either the origin or the end, we need to symmetrically allocate one of two $\ket{1}$ in each component of $\ket{D^x_2}$ to the fourth column of qubits. It can be done by applying the basic module in {\bf Figure~\ref{fig: SplitModule}} in step 3 of {\bf Algorithm~\ref{alg: setPartition}}. The basic module in {\bf Figure~\ref{fig: SplitModule}} has constant gate complexity, and the sets of qubits corresponding to different subsets are mutually disjoint. Every vertex in the same subset belongs to $O(n)$ vertex pairs. Therefore, step 3 requires $O(4\binom{n}{2})=O(n^2)$ gate complexity and $O(n)$ depth.

\begin{figure}[H]
    \centering
    \includegraphics[width=.3\linewidth]{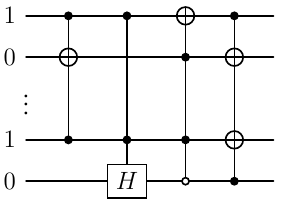}

    \caption{Basic module to split the components in $\ket{D^x_2}$, where $x=a,\ b,\ c,\ d$. We design it to transform the state $\ket{10}\ket{10}$ into the state $\frac{1}{\sqrt{2}}(\ket{10}\ket{01}+\ket{01}\ket{10})$.}
    \label{fig: SplitModule}
\end{figure}

Steps 4-6 of {\bf Algorithm~\ref{alg: setPartition}} encode the information of set partitions. For the set composed of subsets A and B, the first column of qubits is in state $\ket{0}$ and the second column of qubits distinguishes the labels A and B. Therefore, we can prepare the Dicke state $D^{a+b}_b$ over the second column of qubits, which is marked as stage $1$ in {\bf Figure~\ref{fig: SetPartitionTree}}. Similarly, we can prepare the Dicke state $D^{c+d}_d$ over the second column of qubits for subsets C and D, which is marked as stage $2$ in {\bf Figure~\ref{fig: SetPartitionTree}}; and prepare the Dicke state $D^{a+b+c+d}_{c+d}$ over the first column of qubits for all four subsets, which is marked as stage $3$.
When we apply the quantum gates of the Dicke state preparation circuit to a certain qubit, we must ensure that the other qubits encoding the same vertex are correspondingly acted upon. Hence, we need to modify the basic structures in {\bf Figure~\ref{fig: dicke-structure}}. The modified structures are provided in {\bf Figure~\ref{fig: par-structure}}, where we use controlled-SWAP gates to ensure that the information of other qubits follows the qubit acted upon by $RY$ gates. These steps require a gate complexity of $O(3n^2)=O(n^2)$ and a depth of $O(n)$ because the modification of the basic structures only increases the gate complexity by a constant factor compared with~\cite{bartschi2019deterministic}.

\begin{figure}[H]
    \centering
    \begin{subfigure}{0.45\textwidth}
        \centering
        \includegraphics[width=\textwidth]{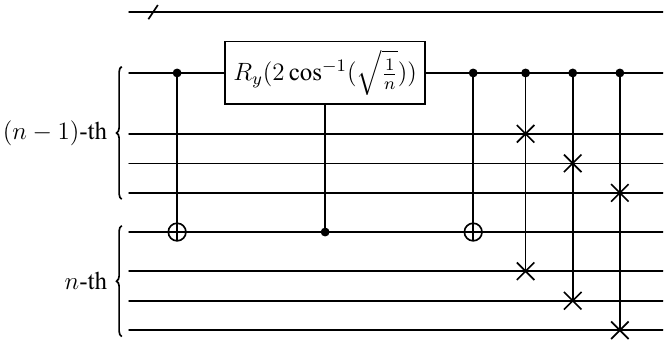}
        \caption{The first modified structure.}
        \label{fig: par-sub1}
    \end{subfigure}
    \hfill
    \begin{subfigure}{0.45\textwidth}
        \centering
        \includegraphics[width=\textwidth]{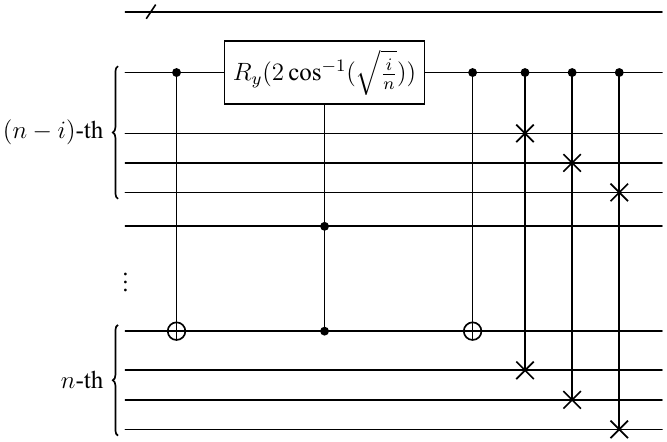}
        \caption{The second modified structure.}
        \label{fig: par-sub2}
    \end{subfigure}
    \caption{Basic modified structures of the quantum circuits for preparing the set partition states, where a vertex is encoded by $4$ qubits and in this figure the set composed of the first qubit of every vertex is acted upon by the $RY$-gates in the Dicke state preparation circuit.}
    \label{fig: par-structure}
\end{figure}

Finally, we demonstrate that we can obtain the set partition state correctly because the process encoding set partitions does not destroy the information about origins and ends:

{\bf Figure~\ref{fig: SetPartitionTree}} shows similar binary trees as {\bf Figure~\ref{fig: DickeTree}}. After the stage $1$ finished, there is a binary tree with $\binom{a+b}{a}$ leaf nodes. Since the state space factorizes as a tensor product, each leaf node connects to a new tree with $\binom{c+d}{c}$ leaf nodes after stage~2 finishes. We obtain a tree with $\binom{a+b}{a}\binom{c+d}{c}$ leaf nodes. After stage~3 finishes, each leaf node of this tree connects to yet another tree with $\binom{a+b+c+d}{a+b}$ leaf nodes. Finally, we obtain a tree with $\binom{a+b}{a}\binom{c+d}{c}\binom{a+b+c+d}{a+b}=\frac{(a+b+c+d)!}{a!b!c!d!}$ leaf nodes. Each path from the root to a leaf in this tree defines a mapping that sends the initial state to a particular component of the final state. This mapping, composed of a series of swap and identity, is a permutation. Since the permutation is bijective, the process encoding set partitions does not destroy the information about origins and ends: All the $2^4\binom{a}{2}\binom{b}{2}\binom{c}{2}\binom{d}{2}$ marked $u,\ v$ pairs are simply permuted. Hence, we get $2^4\binom{a}{2}\binom{b}{2}\binom{c}{2}\binom{d}{2}\frac{(a+b+c+d)!}{a!b!c!d!}$ components corresponding to all possible origins $u$, ends $v$, and $4$-partitions with sizes of $a,\ b,\ c,\ d$.

Step 2 of {\bf Algorithm~\ref{alg: setPartition}} requires the gate complexity of $O(4\cdot2n)=O(n)$.
Step 3 of {\bf Algorithm~\ref{alg: setPartition}} requires the gate complexity of $O(4\binom{n}{2})=O(n^2)$ and the depth of $O(n)$.
Steps 4-6 of {\bf Algorithm~\ref{alg: setPartition}} require the gate complexity of $O(3n^2)=O(n^2)$ and the depth of $O(n)$.
Therefore, our algorithm requires gate complexity $O(n^2)$ and depth $O(n)$ in total. 
There is a trade-off between the use of ancillary qubits and the number of elementary gates for decomposing controlled gates~\cite{gidney2015constructing3}. However, since all control qubits of the controlled gates are of constant order, this trade-off changes the gate complexity only by a constant factor and requires at most a constant number of additional qubits.
\end{proof}

We also consider methods for preparing the uniform superposition states that encode Hamiltonian cycles and permutations of order $n$. Although our algorithm does not use these two superposition states, they are relevant to quantum search algorithm for solving TSP~\cite{zhu2022realizable, bai2025quantum}. (In the previous work~\cite{bai2025quantum}, the authors utilized the superposition state of Hamiltonian cycles to index data with a shortcut of QFT. In this work, we adopt the set partition state instead.) 
Actually, they are quantum counterparts of the classical random generation of cycles and permutations, which are Sattolo's algorithm~\cite{sattolo1986algorithm, wilson2005overview} and Fisher-Yates shuffle~\cite{fisher1957statistical} respectively. Our quantum counterpart of Sattolo's algorithm is a small extension of the previous work~\cite{bai2025quantum}. Here, we improve the gate complexity from $O(n^{\frac{5}{2}})$ in~\cite{bai2025quantum} to $O(n^2\log_2n)$ by optimizing the order of adding vertices and execution of basic modules. Furthermore, we generalize it to the case of permutations. One can refer to the following theorem and its proof in \hyperref[app3]{\bf Appendix C} for details.
\begin{theorem} \label{thm: qdp-HC-P}
    The superposition state of all Hamiltonian cycles of length $n$ $\sum_{\sigma_{n}\in HC_{n}}\frac{1}{\sqrt{(n-1)!}}\ket{\sigma_{n}}$ can be prepared with the gate complexity of $O(n^2\log_2n)$ using $\lceil \log_2n\rceil+1$ ancillary qubits. Similarly, we can prepare the superposition state of all permutations of length $n$ using the same resources.
\end{theorem}

In the case of permutation, there is related work proposing the quantum Fisher-Yates shuffle~\cite{binkowski2025quantum}, which shares the same underlying methods as our algorithm, i.e., quantizing classical random generation algorithms to prepare the superposition states.

\subsection{Implementing oracle with QRAM} \label{secIIIC}
We first prepare the set partition state 
$\sum_{p\in P}\frac{1}{\sqrt{|P|}}\ket{p}=\sum_{p\in P}\frac{1}{\sqrt{|P|}}\otimes_{i=0}^{n-1}\ket{l_io_ie_i}$ using {\bf Algorithm~\ref{alg: setPartition}}.

{\bf Compress the QRAM address space.}
A naive encoding of the set partition state would store the full \(4\)-bit encoding of every vertex, requiring at least \(4n\) qubits and hence a QRAM address space of size \(O(2^{4n})\). This is far larger than the amount of classical data that needs to be accessed, which is \(O^*(2^{nH(\alpha_1)})=O^*(1.865\dots^n)\).
To reduce this overhead, we encode each subset by a natural
number~\cite{siddique2016proofbijectioncombinatorialnumber}. Concretely, for a subset \(S\subseteq[n]\) of size \(s\), we assign the unique index
\[
R(S)=\sum_{j=1}^{s}\binom{x_j}{j},\qquad 0\le R(S)<\binom{n}{s},
\]
where \(x_1<\cdots<x_s\) are the elements of \(S\). The \(\binom{n}{s}\) possible indices \(0,\dots,\binom{n}{s}-1\) are in bijection with all \(s\)-element subsets of \([n]=\{0,1,\dots,n-1\}\). 
Define $a_i(y)=1$ if $l_i=y$ for vertex $i$, 0 otherwise ($l_i,y\in\{00,01,10,11\}$ are labels of subsets); and $\mathbf{1}_{a_i(y)=1\land o_i=1}=1$ if $a_i(y)=1\land o_i=1$ is true, 0 otherwise. Let $t_1=\lceil\log_2(\lceil\alpha_1n\rceil+1)\rceil,t_2=\lceil\log_2\binom{n}{\lceil\alpha_1n\rceil}\rceil+1,t_3=\lceil\log_2n\rceil$. 
{\bf Algorithm~\ref{alg: compressEncoding}} provides the detailed encoding steps. 
In step 7, we utilize the $U_G(\theta)$ operator (a shortcut of QFT) shown in \hyperref[app2]{\bf Appendix B} to realize the quantum adder because one operand of the addition is a combinatorial number that needs to be computed.

\begin{algorithm}[H]
\begin{algorithmic}[1]
\item[\ \ \ \ \textbf{Input:}] Initial state $(\otimes_{i=0}^{n-1}\ket{l_io_ie_i}_1)\ket{0}_{2}\ket{0}^{\otimes t_1}_3\ket{+}^{\otimes t_2}_{4}\ket{0}^{\otimes t_3}_{5}\ket{0}^{\otimes t_3}_6$. The subscript $h$ of $\ket{*}_h$ denotes register position, which is corresponding to the control qubits and target qubits of quantum gates.
\item[\ \ \ \ \textbf{Output:}] $(\otimes_{i=0}^{n-1}\ket{l_io_ie_i}_1)\ket{0}_{2}\ket{s_y}_3\ket{R(S_{p(y)})}_4\ket{u_{p(y)}}_5\ket{v_{p(y)}}_6$, where $s_y$ is the size of the subsets with label $y$.
\STATE We classically precompute all $O(n^2)$ combinatorial numbers $\binom{i}{j}$ for $1\leq j\leq n,0\leq i<n$ ($\binom{i}{j}=0$ if $j>i$), which requires $O(n^3)$ classical operations totally.
\FOR {$i=0$ to $n-1$}
    \STATE Apply $C_1X_2$ ($X$ gate acting on $\ket{*}_2$ controlled by $\ket{*}_1$) to transform the state $\ket{l_io_ie_i}_1\ket{0}_2$ into $\ket{l_io_ie_i}_1\ket{a_i(y)}_2$, where the control qubits in $\ket{*}_1$ match the condition $l_i=y$.
    \STATE Apply a quantum adder to transform the state $\ket{a_i(y)}_2\ket{*}_3$ into $\ket{a_i(y)}_2\ket{*+a_i(y)}_3$.
    \STATE Apply $C_1C_2X_5(i)$ to $\ket{l_io_ie_i}_1\ket{a_i(y)}_2\ket{*}_5$. $C_1C_2X_5(i)$ can transform the state $\ket{l_io_ie_i}_1\ket{a_i(y)}_2\ket{0}_5$ into $\ket{l_io_ie_i}_1\ket{a_i(y)}_2\ket{i\cdot\mathbf{1}_{a_i(y)=1\land o_i=1}}_5$, where $C_1C_2$ matches the condition $a_i(y)=1\land o_i=1$ in $\ket{*}_{1,2}$ and $X_5(i)$ flips the appropriate qubits in $\ket{0}_5$ to transform it into $\ket{i}$. (If $X_5(i)$ is executed, the register $\ket{*}_5$ must be in the state $\ket{0}_5$ because the origin in one subset is unique.)
    \STATE Similar to Step 5, apply $C_1C_2X_6(i)$ to $\ket{l_io_ie_i}_1\ket{a_i(y)}_2\ket{*}_6$. $C_1C_2X_6(i)$ can transform the state $\ket{l_io_ie_i}_1\ket{a_i(y)}_2\ket{0}_6$ into $\ket{l_io_ie_i}_1\ket{a_i(y)}_2\ket{i\cdot\mathbf{1}_{a_i(y)=1\land e_i=1}}_6$.
    \STATE We use precomputed combinatorial numbers as the parameters of controlled $C_2C_3U_G(\theta)_4$ gates where the control qubit in $\ket{*}_2$ ensures that $U_G(\theta)$ is executed only if $a_i(y)=1$ and the control qubits in $\ket{j}_3=\ket{*+a_i(y)}_3$ select $U_G(2\pi\binom{i}{j}/2^{t_2})$ gate to act on $\ket{*}_4$, requiring $i$ controlled $U_G(\theta)$ gates in total (to match all possible $j\leq i$).
    \STATE Apply $C_1X_2$ to uncompute $\ket{*}_2$.
\ENDFOR
\STATE Apply $\text{QFT}^{\dagger}$ to $\ket{*}_4$ to obtain $\ket{R(S_{p(y)})\mod{2^{t_2}}}_4=\ket{R(S_{p(y)})}_4$, since the chosen register width ensures $0\leq R(S_{p(y)})<2^{t_2}$. (see \hyperref[app2]{\bf Appendix B})
\end{algorithmic}
\caption{The algorithm to compute all ranks \(R(S_{p(y)})\) of the $y$-labeled subsets over all set partitions $p$ and get the origin \(u_{p(y)}\) and end \(v_{p(y)}\) of the subset \(S_{p(y)}\). The wildcard $*$ denotes the original content of the register when the corresponding step is executed.}\label{alg: compressEncoding}
\end{algorithm}
 
Step 4 requires $O(\log_2n)$ gate complexity and $O(\log_2\log_2n)$ depth with $O(\log_2n)$ ancillary qubits because there are $O(\log_2n)$ qubits in $\ket{*}_3$ and a quantum adder with operands of $t$ qubits can be implemented in $O(t)$ gate complexity and $O(\log_2t)$ depth using $O(t)$ ancillary qubits~\cite{remaud2025quantum}. 
Steps 5-6 consist of $O(\log_2n)$ $X$ gates acting on $\ket{*}_{5,6}$ controlled by $\ket{*}_{1,2}$, requiring $O(\log_2n)$ gate complexity and depth because there are 5 and $O(\log_2n)$ qubits in $\ket{*}_{1,2}$ and $\ket{*}_{5,6}$, respectively. 
Step 7 includes $i$ controlled $U_G(\theta)$ gates. Therefore, according to \hyperref[app2]{\bf Appendix B}, step 7 requires $O(i\log_2(\alpha_1n)\log_2\binom{n}{\alpha_1n})=O(n^2\log_2n)$ gate complexity and $O(n^2)$ depth with $O(\log_2n)$ ancillary qubits because there are $\lceil\log_2(\alpha_1n)\rceil$ and $\lceil\log_2\binom{n}{\lceil\alpha_1n\rceil}\rceil$ qubits in $\ket{*}_3$ and $\ket{*}_4$, respectively. 
Step 10 requires $O(\log_2^2\binom{n}{\alpha_1n})=O(n^2)$ gate complexity and $O(n)$ depth. 
Therefore, we can realize an operator $\text{compress}(y)$ to prepare the quantum state 

\[\text{compress}(y)\sum_{p\in P}\frac{1}{\sqrt{|P|}}\ket{p}\ket{0}^{\otimes m}=\sum_{p\in P}\frac{1}{\sqrt{|P|}}\ket{p}_1\ket{s_y}_3\ket{R(S_{p(y)})}_4\ket{u_{p(y)}}_5\ket{v_{p(y)}}_6\]
with the gate complexity of $O(n^3\log_2n)$, depth of $O(n^3)$ and $O(\log_2n)$ ancillary qubits, where $s_y$ is the size of the subsets with label $y$ and $m=\lceil\log_2(\lceil\alpha_1n\rceil+1)\rceil+\lceil\log_2\binom{n}{\lceil\alpha_1n\rceil}\rceil+1+2\lceil\log_2n\rceil$.

{\bf Access the classical data.}
We denote $f(S_{p(y)},u_{p(y)},v_{p(y)})$ as $f_r$ ($r=A,B,C,D$ are labels of subsets corresponding to $y=00,01,10,11$, omitting $p$), representing the length of the shortest path within $S_{p(y)}$ from $u_{p(y)}$ to $v_{p(y)}$.
We can realize an operator $\text{select}(f_r)$ through QRAM to prepare the quantum state 
\[
\text{select}(f_r)\ket{R(S_{p(y)})}\ket{u_{p(y)}}\ket{v_{p(y)}}\ket{0}^{\otimes M}=\ket{R(S_{p(y)})}\ket{u_{p(y)}}\ket{v_{p(y)}}\ket{f_r}
\]
with circuit depth $O(\log_2(n^2\binom{n}{\alpha_1n}))=O(0.899\dots n)$ and $O^*(2^{nH(\alpha_1)})=O^*(1.865\dots^n)$ ancillary qubits using only single- and two-qubit gates~\cite{zhang2022quantum, hann2019hardware, hann2021resilience, giovannetti2008quantum, shen2026bucket}, where $M$ depends on the storage accuracy of $f_r$.

We invoke the $\text{compress}(00),\text{select}(f_A),\text{compress}(00)^{\dagger}$ in sequence to transform the state 

\[
\sum_{p\in P}\frac{1}{\sqrt{|P|}}\ket{p}\ket{0}^{\otimes m}\ket{0}^{\otimes4M} \ \ \text{into}\ \ 
\sum_{p\in P}\frac{1}{\sqrt{|P|}}\ket{p}\ket{0}^{\otimes m}\ket{f_A}\ket{0}^{\otimes3M}.
\]
Applying the same operators for $y=01,10,11$ and $r=B,C,D$, we obtain the quantum state 
\[
\sum_{p\in P}\frac{1}{\sqrt{|P|}}\ket{p}\ket{0}^{\otimes m}\ket{f_A}\ket{f_B}\ket{f_C}\ket{f_D}.
\]
Then we use a quantum adder to obtain 
\[
\sum_{p\in P}\frac{1}{\sqrt{|P|}}\ket{p}\ket{0}^{\otimes m}\ket{f_A+f_B+f_C+f_D}\ket{f_B}\ket{f_C}\ket{f_D}, 
\] 
requiring $O(M)$ gate complexity and $O(\log_2M)$ depth~\cite{remaud2025quantum}. Next, we load the weight of the edges connecting different subsets and threshold by applying QFT, controlled $U_G(\omega_{vu})$, $U_G(-C_T)$ and $\text{QFT}^{\dagger}$. According to \hyperref[app2]{\bf Appendix B}, controlled $U_G(\omega_{vu})$ gates require $O(n^2M)$ gate complexity and $O(nM)$ depth, and $U_G(-C_T)$ requires $O(M)$ gate complexity and $O(1)$ depth. $\text{QFT}$ requires $O(M^2)$ gate complexity and $O(M)$ depth. Now we obtain the state
\begin{align}
    &\sum_p\frac{1}{\sqrt{|P|}}\ket{p}_1\ket{0}^{\otimes m}_2|f(S_A,u_A,v_A)+\omega_{v_Au_B}+f(S_B,u_B,v_B)+\omega_{v_Bu_C}+f(S_C,u_C,v_C) \notag\\
    &+\omega_{v_Cu_D}+f(S_D,u_D,v_D)+\omega_{v_Du_A}-C_T\rangle_3\ket{f_B}\ket{f_C}\ket{f_D}.\notag
\end{align}
Mark the solutions with lengths lower than the threshold $C_T$ by applying the $Z$ gate on the sign qubit of the third register $\ket{*}_3$. Finally, we uncompute all value registers by applying the inverse operators.

We choose the value-register width as $M=1+\lceil\log_2(nW+1)\rceil=O(C+\log_2n)$. This width is sufficient to represent all accumulated tour costs and their differences from the threshold in two's-complement form without overflow. Since the edge weights and thresholds are integers, no numerical rounding is required for the threshold comparison. Including uncomputation, the total circuit depth per oracle query is $O(n^3+nM)$, which is polynomial in $n$ and $C$. In particular, when $W$ is polynomially bounded in $n$, we have $M=O(\log_2n)$, which recovers the oracle depth bound $O(n^3)$.
\begin{figure}[H]
    \centering
    \includegraphics[width=.5\linewidth]{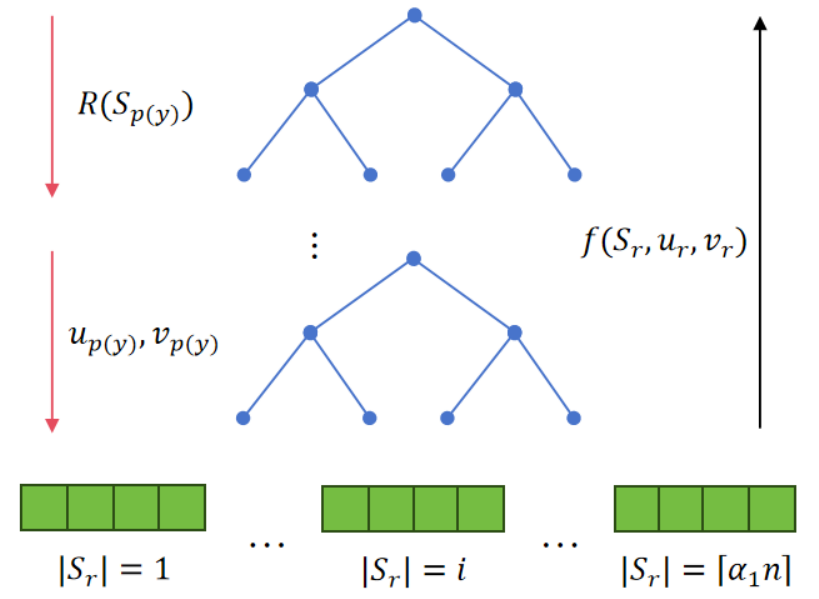}

    \caption{The classical shortest-path data is stored in buckets according to subset size (green boxes). Subsets with the same label $r$ have equal size, and their corresponding data reside in the same bucket. The operator $\text{select}(f_r)$ accesses the bucket associated with $r$ via QRAM. The binary string encoding a certain $(R(S_{p(y)}),u_{p(y)},v_{p(y)})$ determines a path in the binary tree, and the leaf node reached by this path accesses the classical register storing $f(S_r,u_r,v_r)$.}
    \label{fig: QRAM2}
\end{figure}
We briefly explain how $\text{select}(f_r)$ accesses the shortest-path data in the classical DP table based on the components $\ket{R(S_{p(y)})}\ket{u_{p(y)}}\ket{v_{p(y)}}$ of the quantum state using the bucket-brigade QRAM. Classical data is stored in buckets according to subset size. The auxiliary qubits of QRAM are organized into a tree-like logical structure, and the binary string encoding $R(S_{p(y)}),u_{p(y)},v_{p(y)}$ serves as an address that determines an index path from the root node to a leaf node bit by bit.
(For example, when a qubit is in the state $\ket{0}$, the left branch is selected; otherwise, the right branch is selected.) Upon reaching a leaf node, the corresponding classical register is accessed. {\bf Figure~\ref{fig: QRAM2}} illustrates this process.

The QRAM provides coherent access to the classically precomputed DP table. It is one component of the oracle. The oracle additionally computes the required addresses, accumulates path and connecting-edge costs, compares the total cost with the threshold, and uncomputes the work registers. Under the assumed QRAM access model and parallel execution of gates within each circuit layer, each oracle call has polynomial depth. Including classical preprocessing, the overall running time is therefore $O^*(1.866^n)$, with $O^*(1.866^n)$ storage and QRAM hardware resources. This running-time bound does not constitute a bound on the total number of elementary gates.

\section{Simulation experiments} \label{secIV}
\subsection{Framework} \label{secIVA}
To demonstrate the performance of our TSP solver, we conducted simulation experiments on small instances using IBM’s Qiskit. Due to the qubit limitation of the simulator, we replaced the QRAM with the shortcut of QFT in the simulation experiments.
{\bf Figure~\ref{fig: TSPsolver}} shows the framework of our TSP solver in the simulation experiments. We divide all qubits into the index register and the value register. $Ind$-module can prepare the set partition state in the index register. The controlled $U_G$ gate is a shortcut of QFT, which is introduced in \hyperref[app2]{\bf Appendix B}. The oracle queries the classical data based on the index register and calculates the lengths of the corresponding Hamiltonian cycles in the value register. The purple part is standard Grover's diffusion operator $A(I-2\ket{0}\bra{0})A^{\dagger}$ with $A$ replaced by $Ind$-module.

\begin{figure}[H]
    \centering
    \includegraphics[width=0.9\linewidth]{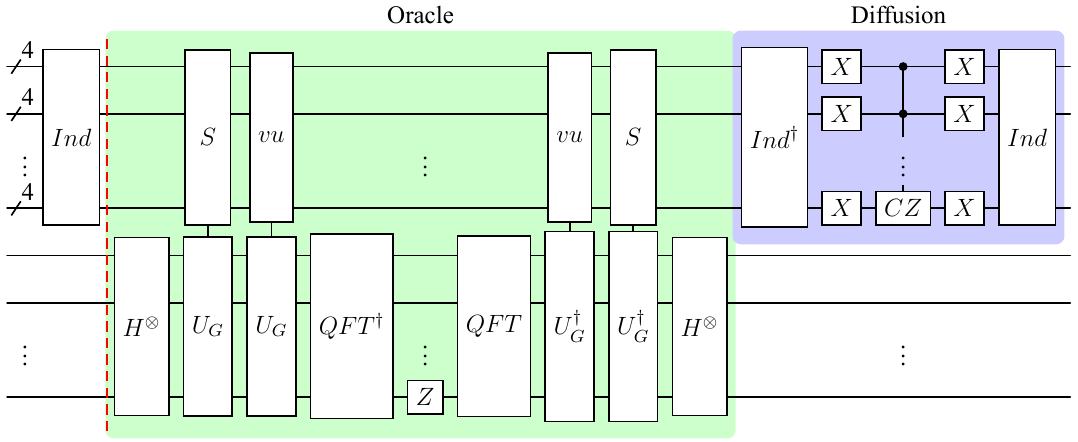}

    \caption{The framework of our TSP solver, where the $Ind$-module represents the set partition state preparation algorithm and the controlled $U_G$ gate is introduced in \hyperref[app2]{\bf Appendix B}. The part to the right of the red dashed line is Grover search module, which will be executed many times to amplify the amplitude of the optimal solution.}
    \label{fig: TSPsolver}
\end{figure}

To reduce the additional resource consumption caused by redundant encoding of solutions, we require that the salesman starts from vertex $0$ and returns to vertex $0$ at the end, i.e., $4$-qubit encoding vertex $0$ is fixed as $\ket{0010}$ and for the other vertices in subset A we set $\ket{u}=\ket{0}$. In this way, we eliminate $n$ repetitions caused by different starting points, and the number of qubits occupied by $Ind$-module will also be reduced from $4n$ to $4(n-1)$.

\subsection{Results} \label{secIVB}
Due to the memory limitation, we use the case of $3$-partition for two graphs with number of vertices $n=6,\ 7$ and adopt shortcut of QFT to load the classical data. The adjacency matrices of these two graphs are: 

\begin{equation*}
    \mathbf{X_6}=\left(
    \begin{array}{cccccc}
        0 & 1 & 3 & 2 & 2 & 1 \\
        1 & 0 & 1 & 3 & 1 & 2 \\
        3 & 1 & 0 & 3 & 2 & 3 \\
        2 & 3 & 3 & 0 & 1 & 1 \\
        2 & 1 & 2 & 1 & 0 & 3 \\
        1 & 2 & 3 & 1 & 3 & 0 \\
    \end{array}\right),\
    \mathbf{X_7}=\left(
    \begin{array}{ccccccc}
        0 & 1 & 3 & 2 & 2 & 1 & 1 \\
        1 & 0 & 1 & 3 & 1 & 2 & 1 \\
        3 & 1 & 0 & 3 & 2 & 3 & 1 \\
        2 & 3 & 3 & 0 & 1 & 1 & 1 \\
        2 & 1 & 2 & 1 & 0 & 3 & 1 \\
        1 & 2 & 3 & 1 & 3 & 0 & 1 \\
        1 & 1 & 1 & 1 & 1 & 1 & 0 \\
    \end{array}\right).\
    \label{equ: adjacency matrices}
\end{equation*}
We can initialize the threshold with the length of a random solution, and iteratively update the threshold by executing the quantum exponential searching~\cite{durr1996quantum, boyer1998tight}. The simulation experiments found the optimal solution on both graphs. 

We analyzed the probability of obtaining optimal solutions under different number of iterations. 
To characterize amplitude amplification, we evaluate the search circuit at a fixed threshold $C_T = w_{p^*} + 1$, where $w_{p^*}$ is the optimal tour cost. Under this threshold, precisely the optimal candidates are marked. The probabilities reported below characterize this fixed-threshold search stage rather than the success probability of the complete adaptive minimum-finding procedure. 
The same quantum circuit was executed for $shots=1000$ times to obtain the statistical data on the components of the final quantum state. 
{\bf Figure~\ref{fig: succ_p}} in \hyperref[app4]{\bf Appendix D} shows the probability of obtaining optimal solutions under different number of iterations. 
{\bf Table~\ref{tab: sim}} shows our simulation results on IBM's Qiskit, where ``qubits" denotes the required number of qubits, ``iterations" denotes the optimal number of iterations, $t$ is the number of iterations, $M_{eff}$ is the number of components corresponding to optimal solutions, $N_{eff}=\frac{(n-1)!\cdot(a-1)b(b-1)c(c-1)}{(a-1)!b!c!}$ is the number of components in the superposition, the theoretical probability is computed from the formula $P(t)=\sin^2[(2t+1)\arcsin(\sqrt{\frac{M_{eff}}{N_{eff}}})]$~\cite{grover1996fast}, and the simulated probability is computed from the proportion of the optimal solutions in 1000 samples. For the graph with $n=6$, we take $a=2,\ b=2,\ c=2$. For the graph with $n=7$, we take $a=3,\ b=2,\ c=2$. 
We also provide the bar charts in {\bf Figure~\ref{fig: sim-res}} in \hyperref[app4]{\bf Appendix D}.

\begin{table}[h]
    \centering
    \begin{threeparttable}[b]
        \caption{Fixed-threshold amplitude-amplification results at the optimal Grover iteration count.} \label{tab: sim}
        \setlength{\tabcolsep}{2mm}
        \begin{tabular}{c c c c c c c c}
            \toprule[1.5pt]
            \     &  $n$ & qubits & $M_{eff}$  & $N_{eff}$     & iterations   & theoretical probability & simulated probability \\ \hline
            $X_6$ &  6   & 25     & 2          & 120           & 6            & $98.7\%$                & $98.9\%$     \\
            $X_7$ &  7   & 29     & 4          & 720           & 10           & $100\%$                 & $100\%$    \\
            \bottomrule[1.5pt]
        \end{tabular}
    \end{threeparttable}
\end{table}

The observed probabilities agree with the theoretical amplitude-amplification behavior. These small-scale experiments validate the search-circuit construction; they do not implement the QRAM architecture assumed in the asymptotic running-time analysis.

\section{Discussion} \label{secV}

We have shown that combining classical dynamic programming with non-nested quantum search achieves a quantum advantage on the NP-hard traveling salesman problem. 
Our quantum divide-and-conquer framework balances classical preprocessing against quantum search through the number and sizes of the subsets. 
To reduce the depth of each oracle query to the polynomial level, we use the QRAM, trading an exponential number of qubits for polynomial circuit depth.
Structured states eliminate the need to precompute, beyond the DP table, an additional long array that directly indexes the solutions. This avoids the extra classical overhead that would otherwise erase the quantum advantage.

{\bf Open question.} An important direction is to develop explicit coherent circuit implementations of the nested scheme in~\cite{ambainis2019quantum} and quantify their resource costs under a specified QRAM model, including the construction of the oracle and uncomputation. Our construction provides an explicit non-nested reference for this comparison.

\subsection{Generalization to other combinatorial optimization problems}
The quantum divide-and-conquer strategy is not limited to TSP. Any combinatorial optimization problem whose solutions decompose into independent components can be tackled using the same three-step approach:
\begin{enumerate}
    \item[(i)] Decompose the solutions into components. Partition the universal set into $k$ disjoint subsets according to the problem structure, cutting each global solution into $k$ components.
    \item[(ii)] Precompute component costs classically. Compute the optimal cost of each component via classical dynamic programming over subsets.
    \item[(iii)] Quantum search over partitions. Prepare a superposition state encoding all valid partitions, construct a problem-specific oracle that assembles components into complete solutions and evaluates their costs, and run quantum minimum finding.
\end{enumerate}

Beyond TSP, this applies to set-partition-based or path-based optimization problems.
The framework is general, but each step must be customized to the target problem. 
Three aspects deserve particular attention. First, the superposition state must encode the relevant combinatorial objects (as we construct for set partitions in {\bf Theorem~\ref{thm: setPartitionState}}). 
Second, the oracle must evaluate the correct cost function. Third, the optimal parameter choices---including the number of subsets $k$, the subset-size distribution, and the balance between classical precomputation and quantum search---are inherently problem-dependent. 
A problem fits the framework when its global solution decomposes into components whose costs are precomputable independently. 
Problems with inherently sequential DP recursion, such as scheduling problems (minimizing total weighted completion time under precedence constraints), where the cost of adding an element depends on the entire accumulated set ($f(S \cup \{i\}) = \min\{f(S) + g(S,i)\}$), do not directly fit this structure. Finding whether other DP formulations admit quantum speedups under this approach is an open question.

\pagebreak
\newpage
\beginsupplement
\setcounter{figure}{0}
\renewcommand{\thefigure}{A\arabic{figure}}

\section*{Appendix} \label{appendix}

\subsection*{A. The bucket-brigade QRAM} \label{app1}
The bucket-brigade QRAM architecture has been realized on the superconducting quantum processor for small models~\cite{shen2026bucket}.
{\bf Figure~\ref{fig: QRAM1}} shows the principle of the bucket-brigade QRAM loading $8$ data items. The address bits within the green box are sequentially fed into the QRAM. At each step, a branch is selected based on whether the current bit is 0 or 1, while the subsequent address bits are then passed downward and stored in the green address register. Upon reaching a leaf node, the classical data item is loaded into the blue data register via quantum gates. The data is then either transmitted back along the reverse path of the address traversal or transferred to the target value register $d$ via quantum teleportation. The transmission of address bits and data is implemented using the controlled-SWAP network. For specific details, see~\cite{shen2026bucket}. 

\begin{figure}[H]
    \centering
    \includegraphics[width=.5\linewidth]{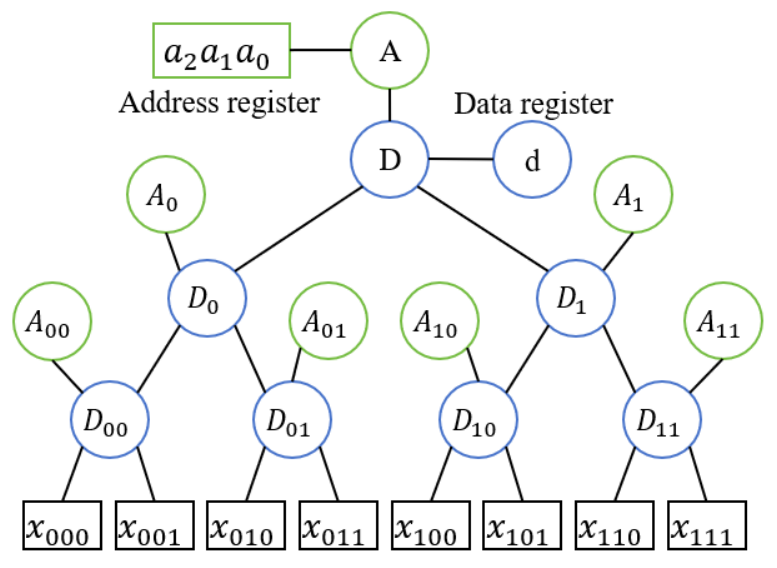}

    \caption{Schematic diagram of the bucket-brigade QRAM loading $8$ data items.}
    \label{fig: QRAM1}
\end{figure}

\subsection*{B. Shortcut of QFT} \label{app2}
In the simulation experiment, we utilize a shortcut of the quantum Fourier transform (QFT) to calculate the data in value register under the control of the index register. Gilliam et al.~\cite{gilliam2021grover} proposed this shortcut to address the constrained polynomial binary optimization (CPBO). We extend this idea to generally handle the value register through the index register.

Define the operator $U_G(\theta)=\bigotimes_{l=0}^{M-1}R(2^l\theta)$, one can verify
\begin{equation}
    U_G(\theta)H^{\otimes M}\ket{0}^{\otimes M}=\frac{1}{\sqrt{2^M}}\sum\limits^{2^M-1}_{j=0}\exp(ij\theta)\ket{j}_M,
    \label{equ:sqft1}
\end{equation}
which is analogous to

\begin{equation}
    QFT\ket{k}=\frac{1}{\sqrt{2^M}}\sum^{2^M-1}_{j=0}\exp(\frac{2\pi ikj}{2^M})\ket{j}_M\ ,\ 0\leq k\leq 2^M-1.
    \label{equ:sqft2}
\end{equation}
Hence, when $\theta$ is fixed, $U_G(\theta)$ can be regarded as a shortcut of QFT because there are no multi-qubit interactions in $U_G(\theta)$. Take $\theta=2\pi k/2^M\ (-2^{M-1}\leq k<2^{M-1})$. Combining (\ref{equ:sqft1}) and (\ref{equ:sqft2}), we have
\begin{equation}
    QFT^{\dagger}U_G(\theta)H^{\otimes M}\ket{0}^{\otimes M}=\ket{k\mod2^M}.
    \label{equ:sqft3}
\end{equation}
For $-2^{M-1}\leq k<0$ we have $\ket{k\mod2^M}=\ket{k+2^M\mod2^M}$, meaning that the highest (leftmost) digit is 1. Therefore, we can easily identify a negative number by examining its sign qubit (the highest digit). Although (\ref{equ:sqft3}) can only encode the integer exactly, there are two methods in~\cite{gilliam2021grover} to handle general real numbers, which are approximating real coefficients by fractions and encoding real coefficients as Fejér distributions, respectively. Now, we can encode data into $\theta$ and operate the value register using $U_G(\theta)$ controlled by index register. Particularly, $U_G(\theta)$ has the property of
\begin{equation}
    U_G(\theta_1)U_G(\theta_2)=U_G(\theta_1+\theta_2),
    \label{equ:sqft4}
\end{equation}
which facilitates the addition of two values. We choose $M$ according to the required encoding accuracy.

In {\bf Figure~\ref{fig: TSPsolver}}, $U_G$ operators controlled by $S$-module are divided into four groups because we want to utilize the localized matching technology to match the subsets with label A, B, C, D separately. For example, we use $2\binom{b}{2}\binom{n-1}{b}$ controlled $U_G(f(S_B,u_B,v_B))$ operators with $4b$ control qubits taking all possible combinations in $4(n-1)$ qubits. (Every $4$ qubits are bound together to encode a vertex.) After we applied all $\binom{a}{1}\binom{n-1}{a}+2\binom{b}{2}\binom{n-1}{b}+2\binom{c}{2}\binom{n-1}{c}+2\binom{d}{2}\binom{n-1}{d}$ controlled $U_G$ operators, we obtain 
\begin{equation}
    QFT\ket{f(S_A,0,v_A)+f(S_B,u_B,v_B)+f(S_C,u_C,v_C)+f(S_D,u_D,v_D)}
\end{equation}
in the value register. Then we utilize $U_G$ operators controlled by $vu$-module to match the edge connections between subsets with different labels. For example, we use $2\binom{n-1}{2}$ controlled $U_G(\omega_{v_Au_B})$ with control qubits matching all possible $\ket{0001},\ \ket{0110}$ pairs located in $4(n-1)$ qubits. After we applied these $\binom{n-1}{1}+3\cdot2\binom{n-1}{2}$ controlled $U_G$ operators, we obtain 
\begin{align}
    &QFT|f(S_A,0,v_A)+\omega_{v_Au_B}+f(S_B,u_B,v_B)+\omega_{v_Bu_C}+f(S_C,u_C,v_C)+\omega_{v_Cu_D} \notag\\
    &\ \ \ \ \ \ \ \ +f(S_D,u_D,v_D)+\omega_{v_D0}\rangle \notag\\
    =&QFT\ket{\text{length of a certain Hamiltonian cycle}}
    \label{equ: HCvalue}
\end{align}
in the value register. Finally, $U_G(-C_T)$ (which is not drawn in {\bf Figure~\ref{fig: TSPsolver}}) and $\text{QFT}^{\dagger}$ act on the value register. Now we can use $Z$ gate acting on the sign qubit to mark a solution with length lower than threshold $C_T$. By applying the inverses of the above operators (except $Z$ gate) we can uncompute the value register.

Through quantum exponential searching~\cite{durr1996quantum, boyer1998tight}, we can find the optimal solution by iteratively updating the threshold.

\subsection*{C. State preparation for Hamiltonian cycle and permutation} \label{app3}
The following two straightforward lemmas enable the iterative generation of Hamiltonian cycles and permutations.
\begin{lem}
    \label{lem: HC-c-generate}
    Given a new vertex $v\notin \{1, \dots,\ k-1\}$, a Hamiltonian cycle of length $k-1$ denoted by
    $\sigma_{k-1}=[\sigma_{k-1}(1),\ \sigma_{k-1}(2),\ \ldots,\ \sigma_{k-1}(k-1)]$ (as a permutation).
    We can generate $k-1$ different Hamiltonian cycles of length $k$ where $1\leq i\leq k-1$:
    \[ \sigma_{k}=\left(
    \begin{array}{ccccccc}
    1& 2& \ldots& i& \ldots& k-1& v\\
    \sigma_{k-1}(1)& \sigma_{k-1}(2)& \ldots& v& \ldots& \sigma_{k-1}(k-1)& \sigma_{k-1}(i)\\
    \end{array}\right).\]
    If we already have all Hamiltonian cycles of length $k-1$, then we can obtain all Hamiltonian cycles of length $k$ using the above method.
\end{lem}

\begin{lem}
    \label{lem: permutation-c-generate}
    Given a new vertex $v\notin \{1, \dots,\ k-1\}$, a permutation of length $k-1$ denoted by
    $\sigma_{k-1}=[\sigma_{k-1}(1),\ \sigma_{k-1}(2),\ \ldots,\ \sigma_{k-1}(k-1)]$.
    we can generate $k$ different permutations of length $k$ where $1\leq i\leq k-1$:
    \[ \sigma_{k}=\left(
    \begin{array}{cccccc}
    1& \ldots& i& \ldots& k-1& v\\
    \sigma_{k-1}(1)& \ldots& v& \ldots& \sigma_{k-1}(k-1)& \sigma_{k-1}(i)\\
    \end{array}\right)\ and\ \ 
    \sigma_{k}=\left(
    \begin{array}{ccccc}
    1& \ldots& k-1& v\\
    \sigma_{k-1}(1)& \ldots& \sigma_{k-1}(k-1)& v\\
    \end{array}\right).\]
    If we already have all permutations of length $k-1$, then we can obtain all permutations of length $k$ through the above method.
\end{lem}

We usually take $v$ as $k$, meaning that we add new vertices in ascending order to iteratively generate Hamiltonian cycles and permutations. The iterative processes in {\bf Lemma~\ref{lem: HC-c-generate}} and {\bf Lemma~\ref{lem: permutation-c-generate}} can be quantized to prepare the uniform superposition states that encode Hamiltonian cycles and permutations of order $n$. 

Before proving {\bf Theorem~\ref{thm: qdp-HC-P}}, we also need the following lemma.
\begin{lem}
    \label{lem: oplus}
    Let $m=\lceil\log_2n\rceil$, $x=x_{m-1}x_{m-2}\dots x_0$ and $y=y_{m-1}y_{m-2}\dots y_0$ be binary representations of positive integers $x,\ y$. Define the following binary operation:
    \begin{equation}
        x\oplus_0y=\left\{
        \begin{array}{rl}
             0, &  \text{if}\ \  x=y\\
             x, &  \text{if}\ \  x\neq y 
        \end{array}
        \right.\ .
        \label{equ: oplus1}
    \end{equation}
    We can borrow $m+1$ auxiliary qubits to transform the state $\sum_{x,y}c_{xy}\ket{x}_1\ket{y}_2\ket{0}_a\ket{0}_b$ into the state $\sum_{x,y}c_{xy}\ket{x\oplus_0y}\ket{y}\ket{0}_a\ket{0}_b$ with the gate complexity of $O(m)$, where $\ket{0}_a$ represents $m$ ancillary qubits and $\ket{0}_b$ represents $1$ extra ancillary qubits, $\sum_{x,y}|c_{xy}|^2=1$.
\end{lem}
\begin{proof}
    For convenience, we omit the summation about $c_{xy}$ and consider $XOR$ operation $\oplus$ and addition $+$ in $\mathbb{F}_2$ as equivalent. Perform the following operations (\ref{equ: oplus2}, \ref{equ: oplus7}-\ref{equ: oplus8}) in sequence:
    \begin{align}
        \prod_{i=0}^{m-1}C_{1i}X_{ai}\prod_{j=0}^{m-1}C_{2j}X_{aj}\prod_{k=0}^{m-1}X_{ak}\ket{x}_1\ket{y}_2\ket{0}_a\ket{0}_b
        ={}&\ket{x}\ket{y}\otimes(\bigotimes_{i=0}^{m-1}\ket{1+x_i+y_i})\ket{0}_b\ ,
        \label{equ: oplus2}
    \end{align}
    where the subscript in the controlled $X$ gate represents the position of the corresponding qubit. For example, $\bm{\bar{C}}_1^mC_{2j}X_{aj}$ is a controlled $X$ gate where anti-control qubits are $m$ qubits in the first register, control qubit is the $j$-th qubit in the second register and the target qubit is the $j$-th qubit in the ancillary register $\ket{0}_a$.
    Define
    \begin{equation}
        \mathbf{1}_{x=y}\equiv\prod_{i=0}^{m-1}(1+x_i+y_i)=\left\{
        \begin{array}{rl}
             1, &  \text{if}\ \  x=y\\
             0, &  \text{if}\ \  x\neq y 
        \end{array}
        \right.\ .\label{equ: oplus4}
    \end{equation}
    
    \begin{align}
        &\prod_{i=0}^{m-1}C_{1i}X_{ai}\prod_{j=0}^{m-1}C_{2j}X_{aj}\prod_{k=0}^{m-1}X_{ak}\cdot C^m_aX_{b}\ket{x}_1\ket{y}_2(\bigotimes_{l=0}^{m-1}\ket{1+x_l+y_l}_a)\ket{0}_{b} \notag\\
        =&\prod_{i=0}^{m-1}C_{1i}X_{ai}\prod_{j=0}^{m-1}C_{2j}X_{aj}\prod_{k=0}^{m-1}X_{ak}\ket{x}_1\ket{y}_2(\bigotimes_{l=0}^{m-1}\ket{1+x_l+y_l}_a)\ket{\mathbf{1_{x=y}}}_{b}=\ket{x}\ket{y}\ket{0}_a\ket{\mathbf{1_{x=y}}}_b\ ,
        \label{equ: oplus7}
    \end{align}
    
    \begin{align}
        &\bm{\bar{C}}_1^mX_{b}\prod_{j=0}^{m-1}C_{b}C_{2j}X_{1j}\ket{x}_1\ket{y}_2\ket{0}_a\ket{\mathbf{1_{x=y}}}_{b} \notag\\
        ={}&\bm{\bar{C}}_1^mX_{b}(\bigotimes_{j=0}^{m-1}\ket{x_j\oplus y_j\cdot\mathbf{1}_{x=y}})\ket{y}\ket{0}_a\ket{\mathbf{1_{x=y}}}_{b}=\bm{\bar{C}}_1^mX_{b}\ket{x\oplus_0y}_1\ket{y}_2\ket{0}_a\ket{\mathbf{1_{x=y}}}_{b} \notag\\
        ={}&\ket{x\oplus_0y}\ket{y}\ket{0}_a\ket{\mathbf{1_{x=y}\oplus\mathbf{1}_{x=y}}}=\ket{x\oplus_0y}\ket{y}\ket{0}_a\ket{0}_{b}\ .
        \label{equ: oplus8}
    \end{align}
    Collecting all the above operations, this algorithm requires $2m$ $X$ gates, $4m$ $CX$ gates, $m$ controlled $X$ gates with 2 control qubits and $2$ controlled $X$ gates with $m$ control qubits. By results about the decomposition of $MCX$ gate~\cite{gidney2015constructing}, this algorithm is implementable with the gate complexity of $O(m)$.
\end{proof}

Proof of {\bf Theorem~\ref{thm: qdp-HC-P}}:
\begin{proof}
    Let $m=\lceil\log_2n\rceil$. We need $nm$ qubits to encode a Hamiltonian cycle $\sigma_n$ of length $n$ into $\ket{\sigma_n}=\bigotimes_{i=0}^{n-1}\ket{\sigma_n(i)}$. We prove the theorem by mathematical induction.
    
    Beginning from $\ket{0}^{\otimes nm}\ket{0}^{\otimes(m+1)}_a$, where $\ket{0}^{\otimes(m+1)}_a$ represents ancillary qubits, we can obtain the state $\ket{0}^{\otimes m}\ket{21}\ket{0}^{\otimes(n-3)m}\ket{0}^{\otimes(m+1)}_a$ encoding $HC_2=\{\sigma_2=[2,\ 1]\}$ by applying some appropriate $X$ gates.
    Suppose we have already prepared the state 
    \begin{equation}
        \ket{0}^{\otimes m}\left(\sum_{\sigma_{k-1}\in HC_{k-1}}\frac{1}{\sqrt{(k-2)!}}\bigotimes_{j=1}^{k-1}\ket{\sigma_{k-1}(j)}\right)\ket{0}^{\otimes(n-k)m}\ket{0}^{\otimes(m+1)}_a.
        \label{equ: q-ind-HC1}
    \end{equation}
    Perform the following operations in sequence when $k<n$:
    \begin{align}
        &U_1\left[\ket{0}^{\otimes m}\left(\sum_{\sigma_{k-1}\in HC_{k-1}}\frac{1}{\sqrt{(k-2)!}}\bigotimes_{j=1}^{k-1}\ket{\sigma_{k-1}(j)}\right)\ket{0}^{\otimes(n-k)m}\ket{0}^{\otimes(m+1)}_a\right]\notag\\
        &=\ket{0}^{\otimes m}\left(\sum_{\sigma_{k-1}\in HC_{k-1}}\frac{1}{\sqrt{(k-2)!}}\bigotimes_{j=1}^{k-1}\ket{\sigma_{k-1}(j)}\right)\otimes(G_k\ket{0}^{\otimes m})\otimes\ket{0}^{\otimes(n-k-1)m}\ket{0}^{\otimes(m+1)}_a \label{equ: q-ind-HC2}\\
        &=\ket{0}^{\otimes m}\left(\sum_{\sigma_{k-1}\in HC_{k-1}}\frac{1}{\sqrt{(k-2)!}}\bigotimes_{j=1}^{k-1}\ket{\sigma_{k-1}(j)}\right)\otimes(\sum_{i=1}^{k-1}\frac{1}{\sqrt{k-1}}\ket{i})\otimes\ket{0}^{\otimes(n-k-1)m}\ket{0}_a^\otimes.
        \label{equ: q-ind-HC3}
    \end{align}
    In formula (\ref{equ: q-ind-HC2}), $G_k$ represents the Grover algorithm with zero theoretical failure rate \cite{long2001grover}. By {\bf Lemma~\ref{lem: oplus}}, we can realize the operation
    
    \begin{align}
        &U_2\left[\ket{0}^{\otimes m}\left(\sum_{\sigma_{k-1}\in HC_{k-1}}\frac{1}{\sqrt{(k-2)!}}\bigotimes_{j=1}^{k-1}\ket{\sigma_{k-1}(j)}\right)\otimes(\sum_{i=1}^{k-1}\frac{1}{\sqrt{k-1}}\ket{i})\otimes\ket{0}^{\otimes(n-k-1)m}\ket{0}^{\otimes(m+1)}_a\right]\notag\\
        &=\ket{0}^{\otimes m}\left(\sum_{\sigma_{k-1}\in HC_{k-1}}\sum_{i=1}^{k-1}\frac{1}{\sqrt{(k-1)!}}(\bigotimes_{j=1}^{k-1}\ket{\sigma_{k-1}(j)\oplus_0i})\otimes\ket{i}\right)\ket{0}^{\otimes(n-k-1)m}\ket{0}^{\otimes(m+1)}_a.
        \label{equ: q-ind-HC5}
    \end{align}
    Isolate the following part from (\ref{equ: q-ind-HC5}):
    \begin{equation}
        \sum_{\sigma_{k-1}\in HC_{k-1}}\sum_{i=1}^{k-1}\frac{1}{\sqrt{(k-1)!}}(\bigotimes_{j=1}^{k-1}\ket{\sigma_{k-1}(j)\oplus_0i}_j)\otimes\ket{i}\ket{0}_a\ ,
        \label{equ: q-ind-HC6}
    \end{equation}
    where $\ket{0}_a$ is the $0$-th qubit in $\ket{0}^{\otimes(m+1)}_a$. We need to transform the state $\ket{\sigma_{k-1}(j)\oplus_0i}$ into $\ket{k}$ when $\sigma_{k-1}(j)\oplus_0i=0$:  
    \begin{align}
        &C_j(k)X_a\cdot C_aX_j(k)\cdot \bm{\bar{C}}^m_jX_a\ket{\sigma_{k-1}(j)\oplus_0i}_j\otimes\ket{i}\ket{0}_a \notag\\
        ={}&C_j(k)X_a\cdot C_aX_j(k)\ket{\sigma_{k-1}(j)\oplus_0i}_j\otimes\ket{i}\ket{\mathbf{1}_{\sigma_{k-1}(j)=i}}_a \notag\\
        ={}&C_j(k)X_a\ket{\sigma_{k-1}(j)\oplus_0i+k\cdot\mathbf{1}_{\sigma_{k-1}(j)=i}}_j\otimes\ket{i}\ket{\mathbf{1}_{\sigma_{k-1}(j)=i}}_a \notag\\
        ={}&\ket{\sigma_{k-1}(j)\oplus_0i+k\cdot\mathbf{1}_{\sigma_{k-1}(j)=i}}_j\otimes\ket{i}\ket{0}_a\ ,
        \label{equ: q-ind-HC7}
    \end{align}
    where $\bm{\bar{C}}^m_jX_a$ flips $\ket{0}_a$ when $\ket{\sigma_{k-1}(j)\oplus_0i}_j=\ket{0}_j$, $C_aX_j(k)$ transforms the state $\ket{0}_j$ into $\ket{k}_j$ when the ancillary qubit is in the state $\ket{1}_a$, $C_j(k)X_a$ resets $\ket{1}_a$ to $\ket{0}_a$ by matching the binary representation of $k$ encoded in control qubits $\ket{*}_j$. Performing operations in~(\ref{equ: q-ind-HC7}) from $j=1$ to $j=k-1$ sequentially, we obtain the state 
    \begin{equation}
        \sum_{\sigma_{k-1}\in HC_{k-1}}\sum_{i=1}^{k-1}\frac{1}{\sqrt{(k-1)!}}(\bigotimes_{j=1}^{k-1}\ket{\sigma_{k-1}(j)\oplus_0i+k\cdot\mathbf{1}_{\sigma_{k-1}(j)=i}}_j)\otimes\ket{i}\ket{0}_a. \notag
    \end{equation}
    We can set $i=\sigma_{k-1}(t),\ t\in\{1,\dots,\ k-1\}$ for $\sigma_{k-1}$ is bijective. By {\bf Lemma~\ref{lem: HC-c-generate}}, we have:
    \begin{align}
        &\sum_{\sigma_{k-1}\in HC_{k-1}}\sum_{t=1}^{k-1}\frac{1}{\sqrt{(k-1)!}}(\bigotimes_{j=1}^{k-1}\ket{\sigma_{k-1}(j)\oplus_0\sigma_{k-1}(t)+k\cdot\mathbf{1}_{\sigma_{k-1}(j)=\sigma_{k-1}(t)}}_j)\otimes\ket{\sigma_{k-1}(t)}\ket{0}_a \notag\\
        ={}&\sum_{\sigma_{k-1}\in HC_{k-1}}\sum_{t=1}^{k-1}\frac{1}{\sqrt{(k-1)!}}(\bigotimes_{j=1}^{t-1}\ket{\sigma_{k-1}(j)})\ket{k}(\bigotimes_{j=t+1}^{k-1}\ket{\sigma_{k-1}(j)})\ket{\sigma_{k-1}(t)}\ket{0}_a \notag\\
        ={}&\sum_{\sigma_{k}\in HC_{k}}\frac{1}{\sqrt{(k-1)!}}\ket{\sigma_k}\ket{0}_a.
        \label{equ: q-ind-HC8}
    \end{align}

    When we have obtained the state $\ket{0}^{\otimes m}\sum_{\sigma_{n-1}\in HC_{n-1}}\frac{1}{\sqrt{(n-2)!}}\ket{\sigma_{n-1}}\ket{0}^{\otimes(m+1)}_a$ by induction, we perform the last step (We don't consider $0$ until the last step because {\bf Lemma~\ref{lem: oplus}} requires $x\neq0$):
    
    \begin{align}
        &U_1\left[\ket{0}^{\otimes m}\left(\sum_{\sigma_{n-1}\in HC_{n-1}}\frac{1}{\sqrt{(n-2)!}}\bigotimes_{j=1}^{n-1}\ket{\sigma_{n-1}(j)}\right)\ket{0}^{\otimes(m+1)}_a\right] \notag\\
        &=(G_{n}\ket{0}^{\otimes m})\otimes\left(\sum_{\sigma_{n-1}\in HC_{n-1}}\frac{1}{\sqrt{(n-2)!}}\bigotimes_{j=1}^{n-1}\ket{\sigma_{n-1}(j)}\right)\ket{0}^{\otimes(m+1)}_a \notag\\
        &=(\sum_{i=1}^{n-1}\frac{1}{\sqrt{n-1}}\ket{i})\otimes\left(\sum_{\sigma_{n-1}\in HC_{n-1}}\frac{1}{\sqrt{(n-2)!}}\bigotimes_{j=1}^{n-1}\ket{\sigma_{n-1}(j)}\right)\ket{0}^{\otimes(m+1)}_a.
        \label{equ: q-ind-HC9}
    \end{align}
    By {\bf Lemma~\ref{lem: oplus}}, we can realize the operation
    \begin{align}
        &U_2\left[(\sum_{i=1}^{n-1}\frac{1}{\sqrt{n-1}}\ket{i})\otimes\left(\sum_{\sigma_{n-1}\in HC_{n-1}}\frac{1}{\sqrt{(n-2)!}}\bigotimes_{j=1}^{n-1}\ket{\sigma_{n-1}(j)}\right)\ket{0}^{\otimes(m+1)}_a\right] \notag\\
        &=\sum_{\sigma_{n-1}\in HC_{n-1}}\sum_{i=1}^{n-1}\frac{1}{\sqrt{(n-1)!}}\ket{i}\otimes(\bigotimes_{j=1}^{n-1}\ket{\sigma_{n-1}(j)\oplus_0i})\ket{0}^{\otimes(m+1)}_a \notag\\
        &=\sum_{\sigma_{n-1}\in HC_{n-1}}\sum_{t=1}^{n-1}\frac{1}{\sqrt{(n-1)!}}\ket{\sigma_{n-1}(t)}\otimes(\bigotimes_{j=1}^{n-1}\ket{\sigma_{n-1}(j)\oplus_0\sigma_{n-1}(t)})\ket{0}^{\otimes(m+1)}_a \notag\\
        &=\sum_{\sigma_{n-1}\in HC_{n-1}}\sum_{t=1}^{n-1}\frac{1}{\sqrt{(n-1)!}}\ket{\sigma_{n-1}(t)}\otimes(\bigotimes_{j=1}^{t-1}\ket{\sigma_{n-1}(j)})\ket{0}^{\otimes m}(\bigotimes_{j=t+1}^{n-1}\ket{\sigma_{n-1}(j)})\ket{0}^{\otimes(m+1)}_a.
        \label{equ: q-ind-HC10}
    \end{align}
    By {\bf Lemma~\ref{lem: HC-c-generate}}, we obtain the state
    \begin{equation}
        \sum_{\sigma_{n}\in HC_{n}}\frac{1}{\sqrt{(n-1)!}}\ket{\sigma_{n}}\ket{0}^{\otimes(m+1)}_a.
        \label{equ: q-ind-HC11}
    \end{equation}

    Now, we analyze the gate complexity. $G_k$ in $U_1$ can be regarded as the amplitude amplification method~\cite{brassard2000quantum} and expanded as 
    \begin{equation}
        G_{k}(\phi,\varphi)\ket{0}^{\otimes u}=(H^{\otimes u}S_0(\phi)H^{\otimes u}S_{\chi}(\varphi))^pH^{\otimes u}\ket{0}^{\otimes u}=\sum_{i=1}^{k-1}\frac{1}{\sqrt{k-1}}\ket{i},\ \ u=\lceil\log_2k\rceil,
        \label{equ: AAM1}
    \end{equation}
    \begin{equation}
        S_0(\phi)\ket{j}=\left\{ \begin{aligned}  
            \exp(i\phi)\ket{j},\ j=0\\
            \ket{j},\ j\neq0
        \end{aligned}\right.\ \ ,\ \ \ 
        S_{\chi}(\varphi)\ket{j}=\left\{ \begin{aligned}  
            \exp(i\varphi)\ket{j},\ 1\leq j\leq k-1\\
            \ket{j},\ j=0\ or\ k\leq j<2^u
        \end{aligned}\right.\ .
        \label{equ: AAM2}
    \end{equation}
    We adopt three dimensional rotation form according to \cite{long2001grover}:
    \begin{gather}
        p=\lceil\frac{\pi}{4\arcsin{\sqrt{(k-1)/2^u}}}-\frac{1}{2}\rceil,\ \ \phi=\varphi=2\arcsin{\frac{\sin{\frac{\pi}{4p+2}}}{\sqrt{(k-1)/2^u}}}.
        \label{equ: AAM3}
    \end{gather}
    Both $S_0(\phi)$ and $S_{\chi}(\varphi)$ can be implemented using $O(u)$ gates, via a technique called phase kick-back~\cite{chiew2019graph, lee2016generalised}. $p=O(\sqrt{2^u/(k-1)})=O(2)$ for $u\leq\log_2k+1$ and $k>2$. Hence, $U_1$ requires the gate complexity of $O(pu)=O(\log_2k)$. By {\bf Lemma~\ref{lem: oplus}}, $U_2$ requires the gate complexity of $O((k-1)m)=O(k\log_2n)$. In (\ref{equ: q-ind-HC7}), $\bm{\bar{C}}^m_jX_a$, $C_aX_j(k)$ and $C_j(k)X_a$ require $O(m)$ elementary gates, which yields a complexity of $O((k-1)m)=O(k\log_2n)$ for (\ref{equ: q-ind-HC7}).
    
    Collecting all the above complexities, our algorithm to prepare the uniform superposition state of Hamiltonian cycles of length $n$ requires the gate complexity of $O(\sum_k(\log_2k+k\log_2n+k\log_2n))=O(n^2\log_2n)$.

    Comparing {\bf Lemma~\ref{lem: permutation-c-generate}} with {\bf Lemma~\ref{lem: HC-c-generate}}, we just need to modify two places for case of permutation. One is the initial state, which should be changed to $\ket{0}^{\otimes m}\ket{1}\ket{0}^{\otimes(n-2)m}\ket{0}^{\otimes(m+1)}_a$. The other is $G_k$, which should be changed to
    \begin{equation}
        G_{k}(\phi,\varphi)\ket{0}^{\otimes u}=(H^{\otimes u}S_0(\phi)H^{\otimes u}S_{\chi}(\varphi))^pH^{\otimes u}\ket{0}^{\otimes u}=\left\{ \begin{aligned}  
            \sum_{i=1}^k\frac{1}{\sqrt{k}}\ket{i},\ 2\leq k<n\\
            \sum_{i=0}^{n-1}\frac{1}{\sqrt{k}}\ket{i},\ \ \ \ \ \ k=n
        \end{aligned}\right.\ \ ,
        \notag
    \end{equation}
    where $u=\lceil\log_2(k+1)\rceil$ for $2\leq k<n$ and $u=\lceil\log_2n\rceil$ for $k=n$. The extra component in the superposition state compared to (\ref{equ: AAM1}) precisely contributes to the additional permutation $\sigma_k$ in {\bf Lemma~\ref{lem: permutation-c-generate}}.
    Similar to (\ref{equ: AAM3}), $p=O(\sqrt{2^u/k})=O(2)$ for $u\leq\log_2(k+1)+1$ and $k\geq2$. Therefore, this algorithm maintains the same complexity for the case of permutation.
\end{proof}

\subsection*{D. The simulation results of our TSP solver.}\label{app4}
\begin{figure}[H]
    \centering
    \begin{subfigure}{1.\textwidth}
        \centering
        \includegraphics[width=\textwidth]{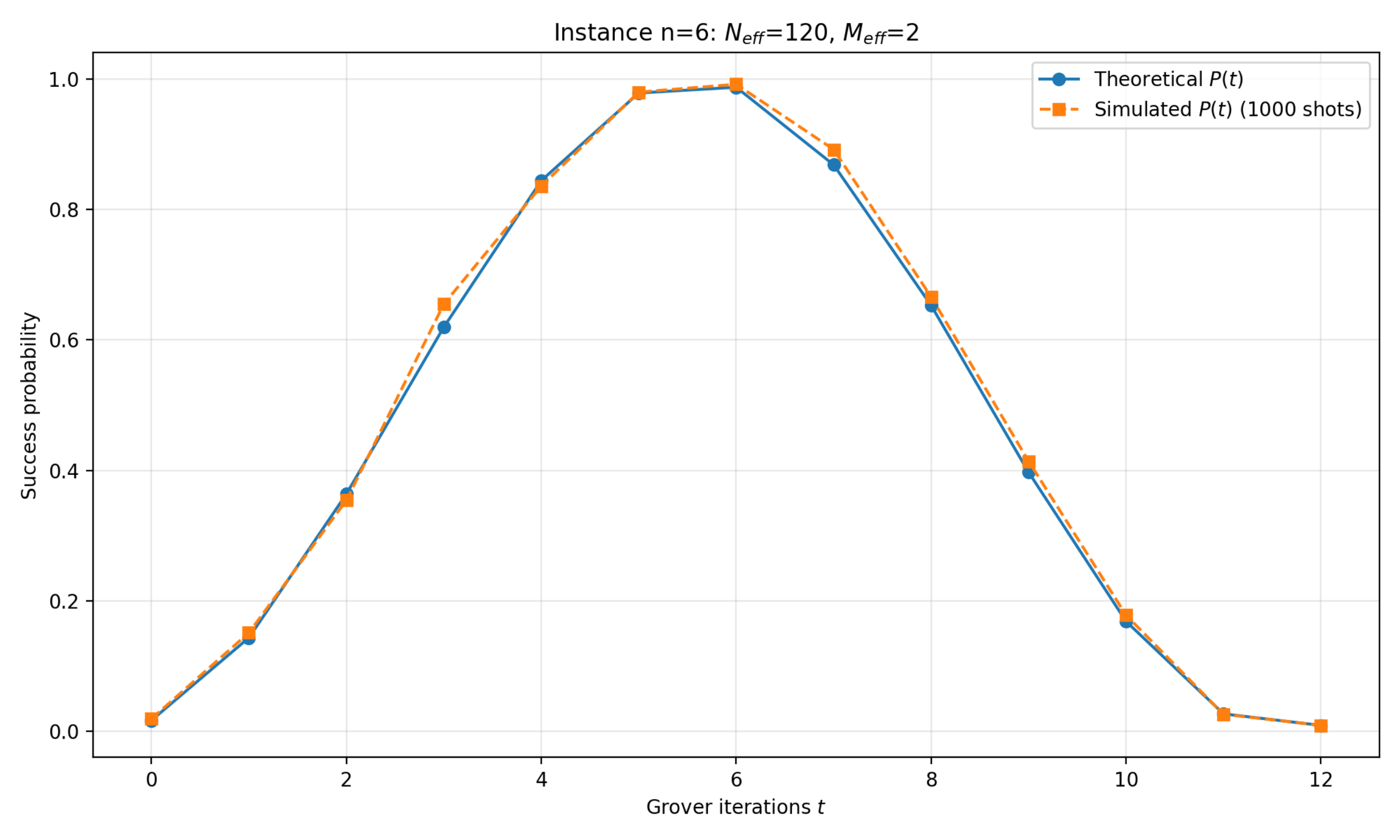}
        \caption{Graph $X_6$.}
        \label{fig: p-sub1}
    \end{subfigure}
    \hfill
    \begin{subfigure}{1.\textwidth}
        \centering
        \includegraphics[width=\textwidth]{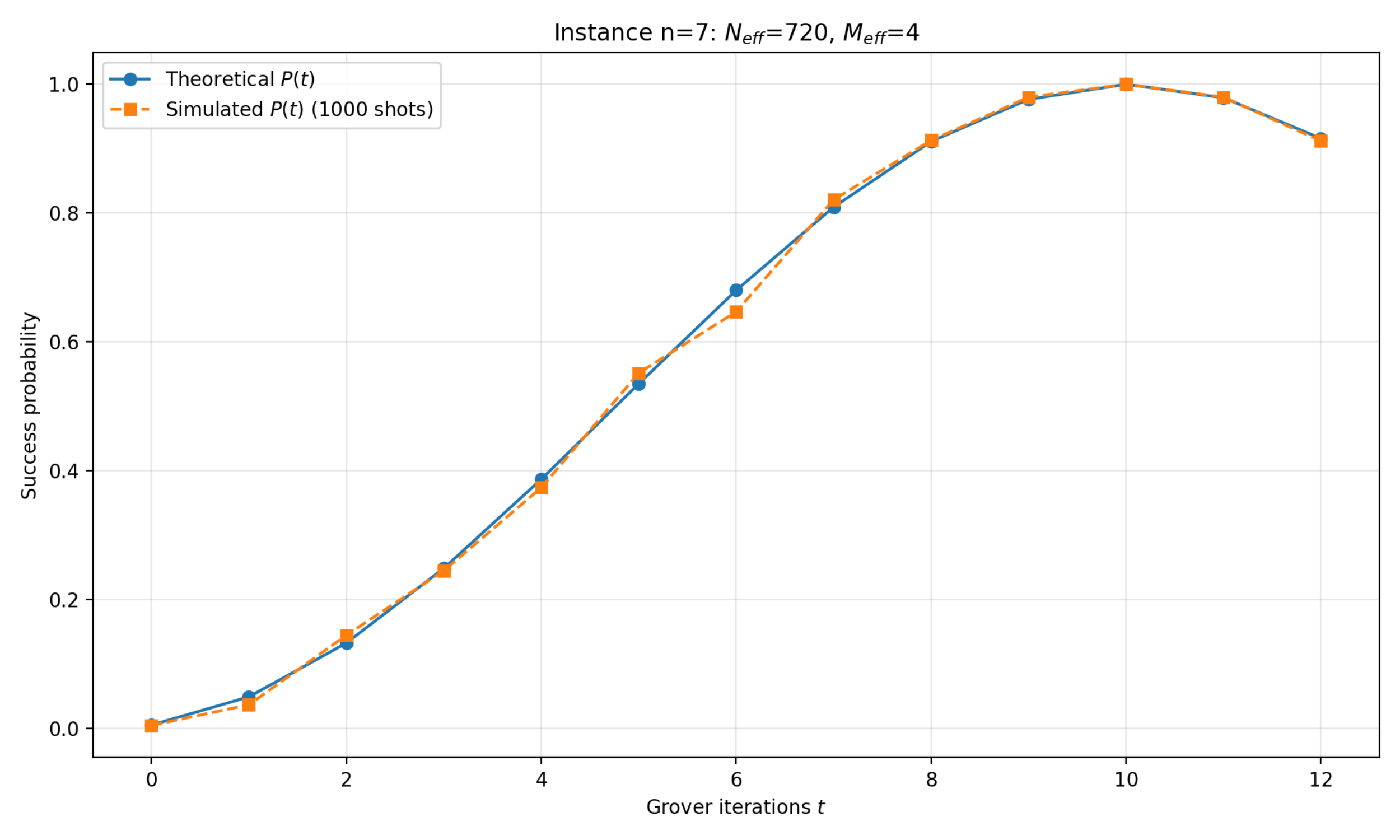}
        \caption{Graph $X_7$.}
        \label{fig: p-sub2}
    \end{subfigure}
    \caption{Success probability versus the number of Grover iterations $t$.}
    \label{fig: succ_p}
\end{figure}

\begin{figure}[H]
    \centering
    \begin{subfigure}{1.\textwidth}
        \centering
        \includegraphics[width=\textwidth]{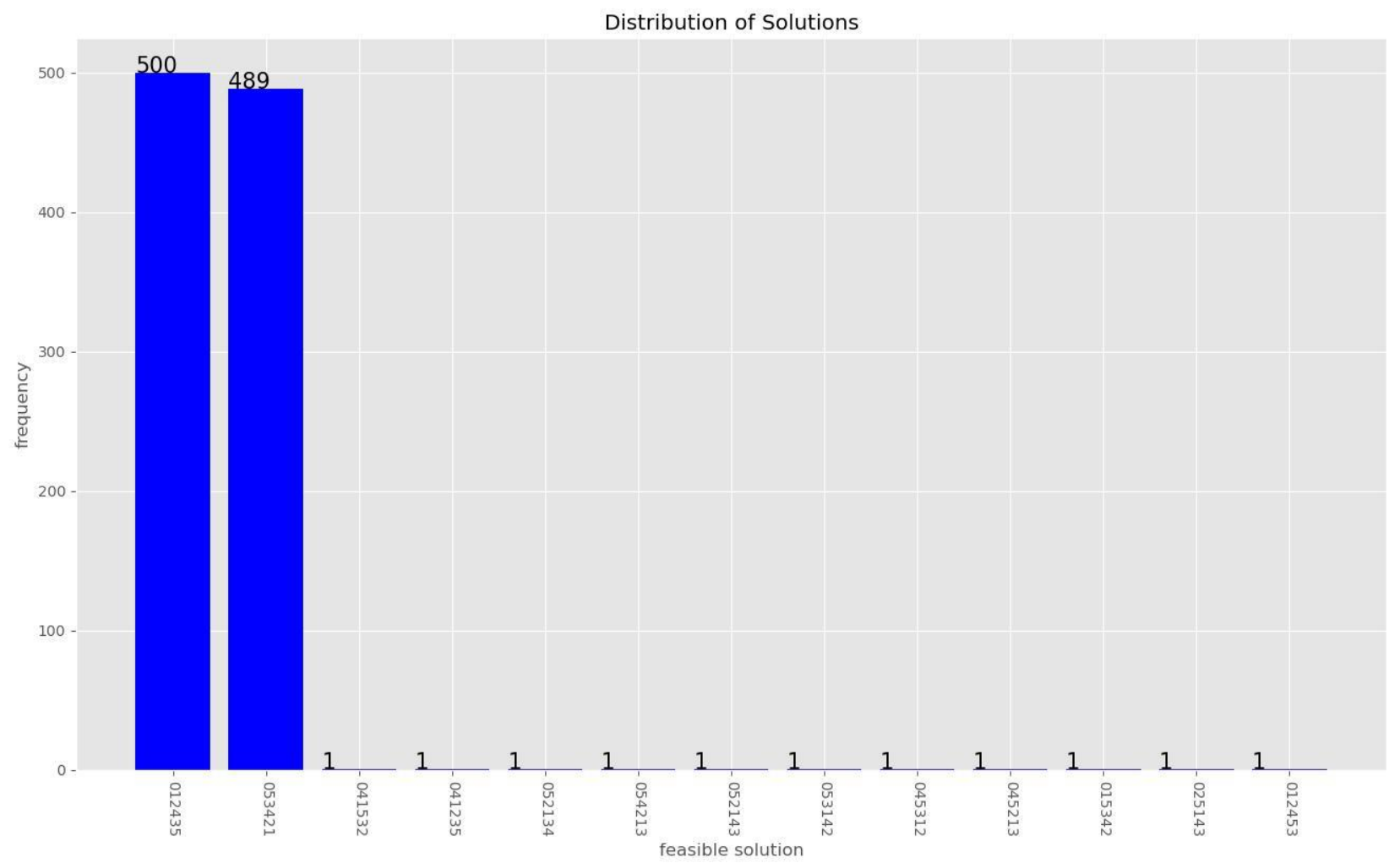}
        \caption{Result of graph with $n=6$.}
        \label{fig: res1}
    \end{subfigure}
    \hfill
    \begin{subfigure}{1.\textwidth}
        \centering
        \includegraphics[width=\textwidth]{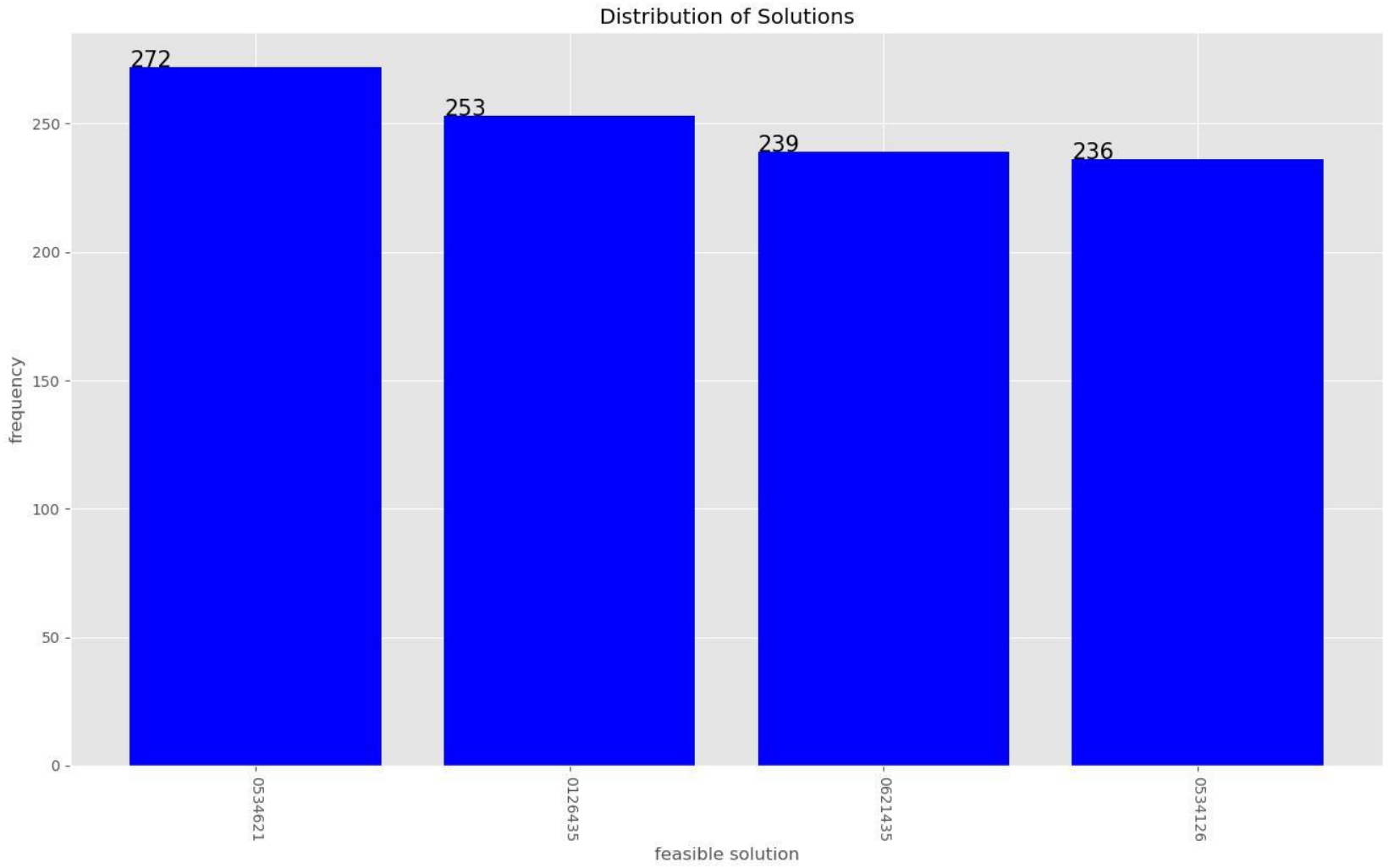}
        \caption{Result of graph with $n=7$.}
        \label{fig: res2}
    \end{subfigure}
    \caption{The simulation verification results of our TSP solver under the optimal number of iterations. The feasible solutions are represented by paths. For example, $012435$ represents a route $0\to1\to2\to4\to3\to5\to0$. The bars with significant proportions are all optimal solutions.}
    \label{fig: sim-res}
\end{figure}

\bibliographystyle{plain}
\bibliography{nature}

\end{document}